\documentclass[11pt]{article}

\usepackage[utf8]{inputenc}
\usepackage{mathpazo}
\usepackage{geometry}
\usepackage{subcaption}
\usepackage{amsmath}
\usepackage{enumerate}
\usepackage{amsfonts}
\usepackage{amssymb}
\usepackage{amsthm}
\usepackage{mathtools}
\usepackage{bm}
\usepackage{xcolor}
\usepackage{hyperref}
\usepackage{enumitem}
\usepackage{csquotes}
\usepackage[ruled, vlined, linesnumbered]{algorithm2e}
\usepackage{multirow}
\usepackage{tikz}
\usetikzlibrary{arrows}
\usetikzlibrary{decorations.pathreplacing,angles,quotes}

\setlist[itemize]{
  noitemsep,
  parsep=0pt,
  topsep=0pt,
  itemsep=1pt,
  partopsep=0pt,
  label=--}
\setenumerate{
  noitemsep,
  parsep=0pt,
  topsep=0pt,
  itemsep=1pt,
  partopsep=0pt}
  
\SetKwInput{KwData}{Input}%
\SetKwInput{KwResult}{Output}%

\theoremstyle{plain}
\newtheorem{theorem}{Theorem}[section]
\newtheorem*{theorem*}{Theorem}
\newtheorem{proposition}[theorem]{Proposition}
\newtheorem*{proposition*}{Proposition}

\newtheorem*{corollary*}{Corollary}
\newtheorem{lemma}[theorem]{Lemma}
\newtheorem*{lemma*}{Lemma}

\theoremstyle{definition}
\newtheorem{definition}[theorem]{Definition}
\newtheorem{remark}[theorem]{Remark}
\newtheorem{example}[theorem]{Example}
\newtheorem*{example*}{Example}

\newtheorem{protocol}[theorem]{Protocol}

\renewcommand\epsilon{\varepsilon}

\newcommand\calA{\mathcal{A}}
\newcommand\calC{\mathcal{C}}
\newcommand\calQ{\mathcal{Q}}

\newcommand\Query{\mathsf{Query}}
\newcommand\Answer{\mathsf{Answer}}
\newcommand\Recover{\mathsf{Recover}}

\newcommand\FF{\mathbb{F}}

\newcommand\PP{\mathbb{P}}
\newcommand\NN{\mathbb{N}}
\renewcommand\AA{\mathbb{A}}

\newcommand\bfa{{\bm a}}
\newcommand\bfc{{\bm c}}
\newcommand\bfe{{\bm e}}

\newcommand\bfi{{\bm i}}
\newcommand\bfk{{\bm k}}
\newcommand\bfq{{\bm q}}

\newcommand\bfw{{\bm w}}
\newcommand\bfx{{\bm x}}
\newcommand\bfy{{\bm y}}
\newcommand\bfz{{\bm z}}
\newcommand\bfzero{{\bm 0}}

\newcommand\mydef{\coloneqq}

\DeclareMathOperator{\Deg}{Deg}
\DeclareMathOperator{\WRM}{WRM}
\DeclareMathOperator{\RS}{RS}
\DeclareMathOperator{\ev}{ev}
\DeclareMathOperator{\wdeg}{wdeg}
\DeclareMathOperator{\Span}{Span}
\DeclareMathOperator{\Lift}{Lift}

\newcommand\Red[1]{{\rm Red}^\star_{#1}}
\newcommand\equivstar[1]{\,\equiv^\star_{#1}\,}

\newcommand\ps[2]{{\left\langle#1,#2\right\rangle}}

\title{
  Weighted Lifted Codes: Local Correctabilities\\
  and Application to Robust Private Information Retrieval
}
\author{
    Julien Lavauzelle\thanks{IRMAR - UMR CNRS 6625, Université de Rennes 1, France. Email: {\tt julien.lavauzelle@univ-rennes1.fr}}
    \and
    Jade Nardi\thanks{Institut de Math\'ematiques de Toulouse~; UMR 5219, Universit\'e de Toulouse~; CNRS UPS IMT, F-31062 Toulouse Cedex 9, France. Email: {\tt jade.nardi@math.univ-toulouse.fr}}
}
\date{\today}

\begin{document}

\maketitle

\begin{abstract}
    Low degree Reed-Muller codes are known to satisfy local decoding properties which find applications in private information retrieval (PIR) protocols, for instance. However, their practical instantiation encounters a first barrier due to their poor information rate in the low degree regime. This lead the community to design codes with similar local properties but larger dimension, namely the lifted Reed-Solomon codes.
    
    However, a second practical barrier appears when one requires that the PIR protocol resists collusions of servers. In this paper, we propose a solution to this problem by considering \emph{weighted} Reed-Muller codes. We prove that such codes allow us to build PIR protocols with optimal computation complexity and resisting to a small number of colluding servers.
    
    In order to improve the dimension of the codes, we then introduce an analogue of the lifting process for weigthed degrees. With a careful analysis of their degree sets, we notably show that the weighted lifting of Reed-Solomon codes produces families of codes with remarkable asymptotic parameters.
\end{abstract}


\section{Introduction}

\subsection{Weighted Reed-Muller codes}

Weighted Reed-Muller codes were introduced by S{\o}rensen in 1992, as a generalisation of Reed-Muller codes in the context of weighted polynomial rings~\cite{Sor}. Formally, given a finite field $\FF_q$, a \emph{weight} $\omega = (\omega_1, \dots, \omega_m) \in (\NN^*)^m$ and a polynomial \[P(X_1, \dots, X_m) = \sum_{\bfi = (i_1, \dots, i_m) \in I} p_\bfi X^{i_1} \dots X^{i_m} \in \FF_q[X_1, \dots, X_m],\] the weighted degree of $P$ with respect to $\omega$ is
\[
    \wdeg_\omega(P) \mydef \max \left\{ \sum_{j=1}^m \omega_j i_j \mid \bfi = (i_1, \dots, i_m) \in I \text{ and } p_\bfi \ne 0 \right\}\,.
\]
In particular, if $\omega = (1, \dots, 1)$, then we get the usual notion of total degree for multivariate polynomials.

In order to build codes from subspaces of polynomials, we consider the evaluation map
\[
    \begin{array}{rclc}
    \ev_{\FF_q^m} :& \FF_q[X_1,\dots,X_m] &\to &\FF_q^{q^m}  \\
                   & P(x_1, \dots, x_m)  & \mapsto &(  P(x_1, \dots, x_m), \bfx = (x_1, \dots, x_m) \in \FF_q^m )
    \end{array}
\]
Then, a weighted Reed-Muller code is defined as the image by $\ev_{\FF_q^m}$ of a subspace of polynomials whose weighted degree is bounded by some integer $d$.

\begin{definition}[Weighted Reed-Muller code]
    Let $m \ge 1$, $\omega \in (\NN^*)^m$ and $d \in \NN$. The \emph{weighted (affine) Reed-Muller code} of order $m$, degree $d$ and weight $\omega$ is:
    \[
        \WRM_q^\omega(d) = \{ \ev_{\FF_q^m}(P), P \in \FF_q[X_1, \dots, X_m], \wdeg_\omega(P) \le d \}\,.
    \]
\end{definition}
Note that weighted Reed-Muller codes are generalised Goppa codes on the weighted projective space $\PP(1,\omega_1,\dots,\omega_m)$ with evaluation points outside the line at infinity $X_0 =0$.

The dimension of weighted Reed-Muller codes, as well as bounds on the minimum distance, are given by S{\o}rensen in his seminal paper~\cite{Sor}. Notice that these parameters are also analysed in a recent work~\cite{AubryCGLOR17} by Aubry, Castryck, Ghorpade, Lachaud, O'Sullivan, and Ram, who also describe minimum weight codewords with geometric techniques. Geil and Thomsen~\cite{GT} finally proved that weighted Reed-Muller codes are efficiently decodable up to half their minimum distance, notably using an embedding of weighted Reed-Muller codes into Reed-Solomon codes.

\subsection{Technical overview and organisation}

In this work, we will only focus on the case where $m = 2$ and $\omega$ is of the form $\omega = (1, \eta)$ where $\eta \ge 1$. This setting seems very restrictive, but it is the most promising in terms of parameters (see for instance~\cite{Sor, GT}) and it also finds a practical application in private information retrieval protocols. For simplicity, we will use the shorter notation $\WRM_q^\eta(d)$ for $\WRM_q^{(1, \eta)}(d)$.

Our first observation is that, when $d \le q-1$, the evaluation map $\ev_{\FF_q^2}$ is injective. This has two major consequences: (i) the code and its parameters are easier to describe and (ii) puncturing the code on \enquote{lines of weighted degree $\eta$} leads to highly-sound local correction. More precisely, in Section~\ref{sec:local-correction} we prove the following result.
\begin{theorem}[informal]
    Let $\eta \ge 1$, $q$ be a prime power and $\gamma \in (0,1)$. For a fixed $\delta \in (0,1)$ small enough, the family of weighted Reed-Muller codes $\WRM^\eta_q( \lfloor \gamma q \rfloor)$ are $(q-1, \delta, \epsilon)$-locally correctable, where $\epsilon = O_\gamma(\delta)$.
\end{theorem}

This result is obtained thanks to the following fact. Let $\phi(T) \in \FF_q[T]$ be a univariate polynomial of (non-weighted) degree bounded by $\eta$, and let $L = ( (t, \phi(t)), t \in \FF_q ) \subset \FF_q^2$. Then for every $\bfc = \ev_{\FF_q^2}(f(X,Y)) \in \WRM_q^\eta(d)$, the restriction $\bfc_{|L}$ of the vector $\bfc$ to the coordinates indexed by elements of $L$ is a codeword of a Reed-Solomon code of degree $d$. Hence, if the codeword $\bfc$ is corrupted with a constant fraction of errors, picking $\phi$ at random and correcting $\bfc_{|L}$ succeeds with constant probability. As a consequence, it allows us to retrieve some symbols of the corrupted codeword in sublinear query complexity.

However, results described above do not improve the related \enquote{local decoding on curves} technique, described for instance by Yekhanin in his survey~\cite{Yekhanin12}. Fortunately, local correctabilities of weighted Reed-Muller codes can be applied to private information retrieval protocols in order to resist collusion of servers.
In particular, we prove that any weighted Reed-Muller code $\WRM^\eta_q(d)$ induces a private information retrieval protocol for databases of $\simeq q^2/2\eta$ entries, requiring a minimal computation complexity for the $q$ servers, and remaining private against any collusion of $\eta$ servers. We refer the reader to Section~\ref{sec:PIR} for more details.

One should notice that the maximal number of entries in the database is directly given by the dimension of $\WRM^\eta_q(d)$. Unfortunately, the information rate of such codes remains bounded by $1/2\eta$ as long as $d \le q-1$, a constraint which is necessary in our context. Therefore, following the seminal paper of Guo, Kopparty and Sudan~\cite{GuoKS13} and subsequent works~\cite{Guo16, Lavauzelle18b}, we initiate the study of a \emph{weighted lifting} of Reed-Solomon codes in order to produce codes with the same local properties as weighted Reed-Muller codes, but with a much larger dimension.

Definitions and essential properties of \emph{weighted lifted codes} are given in Section~\ref{sec:lift}. Similarly to the constructions of lifted (affine~\cite{GuoKS13} and projective~\cite{Lavauzelle18b}) Reed-Solomon codes and lifted Hermitian codes~\cite{Guo16}, we also prove that for fixed $\eta$ and $q \to \infty$, weighted lifts of Reed-Solomon codes are locally correctable with (i) a non-zero asymptotic information rate in the context of errors with constant relative weight, or 
(ii) an information rate arbitrary close to $1$ when errors have smaller weight.



These two results are the main technical outcomes of the paper, and we present them in Section~\ref{sec:degree-set}. They are obtained after a precise analysis of so-called \emph{degree sets} of weighted Reed-Muller and lifted codes, which represent the sets of exponents of monomials spanning the codes.  We finally provide numerical computations of dimensions of weighted lifted codes, which illustrate the improvement of weighted lifted codes over weighted Reed-Muller codes, and their practical useability in private information retrieval.

\section{Local correction of weighted Reed-Muller codes}
\label{sec:local-correction}

\subsection{Restricting Reed-Muller codes to weighted lines}

The local decoding properties of Reed-Muller codes come from the restriction of their codewords on a line being Reed-Solomon codewords. Expecting similar properties on weighted Reed-Muller codes, we have to find what will play the part of the lines in $\PP(1,1,\eta)$.


\begin{definition}[$\eta$-line on $\PP(1,1,\eta)$]
    Let $\eta \ge 1$. We call a (non-vertical) \emph{$\eta$-line} on $\PP(1,1,\eta)$ the set of zeroes of the polynomial $P(X_0,X_1,X_2) = X_2 - \phi(X_0,X_1)$ where $\phi \in \FF_q[X_0,X_1]$ is homogeneous of degree $\eta$.
\end{definition}

Since we evaluate polynomials only at points outside the line $X_0 = 0$, we shall define an $\eta$-line on the affine plane $\AA^2$, viewed as the domain $X_0 \neq 0$, as the intersection of an $\eta$-line on $\PP(1,1,\eta)$ and $X_0 \neq 0$.

\begin{definition}[affine $\eta$-line]
    Let $\eta \ge 1$. We call a (non-vertical) \emph{$\eta$-line} on $\AA^2$ the set of zeroes of a bivariate polynomial $P(X,Y) = Y - \phi(X)$, where $\phi \in \FF_q[X]$ and $\deg \phi \le \eta$.
\end{definition}

Let us remark that if $P = Y - \phi(X)$ defines an $\eta$-line, then $\wdeg_\eta(P) \le \eta$. The converse is not true, since we removed from the definition collections of \enquote{vertical lines} defined by $\phi(X) = 0$, $\deg \phi \le \eta$.

An $\eta$-line can be parametrized by $t \mapsto (t, \phi(t))$. We thus define
\[
    \Phi_\eta = \{ L_\phi : t \mapsto (t, \phi(t)) \mid \phi \in \FF_q[T] \text{ and } \deg \phi \le \eta \}\,,
\]
the set of embeddings of $\eta$-lines into the affine plane $\AA^2 = \overline{\FF_q}^2$. These embeddings are very useful when trying to characterise restrictions of weighted Reed-Muller codes to $\eta$-lines.

\begin{proposition}
    \label{prop:restriction}
    Any polynomial $f \in \FF_q[X,Y]$ whose evaluation over $\FF_q^2$ lies in $\WRM^\eta_q(d)$ satisfies $\ev_{\FF_q}(f \circ L) \in \RS_q(d)$ for any $L \in \Phi_\eta$.
\end{proposition}

\begin{proof}
It is sufficient to check the result on monomials. Let $f=X^iY^j$ where $i+\eta j \leq d$. For every $\phi \in \Phi_\eta$, the univariate polynomial $(f \circ L_\phi)(T) = T^i \phi(T)^j$ has degree less than $d$.
\end{proof}

\subsection{Local correction}

Local decoding was introduced by Katz and Trevisan~\cite{KatzT00} in order to characterise codes allowing to (probabistically) retrieve a message coordinate with a sublinear number of queries in the code length $n$. The difficulty comes from the fact that the retrieval must succeed with non-negligeable probability for \emph{every} codeword which is corrupted by \emph{any} possible error whose weight is bounded by a linear function in $n$. Local correction is very similar to local decoding, the only difference being that one requires that any coordinate of the \emph{codeword} can be retrieved.

Before giving a formal definition of this notion, let us introduce some notation. We denote the Hamming distance between two vectors $\bfx, \bfy$ by $d_H(\bfx, \bfy)$. The weight of $\bfx$ is ${\rm wt}(\bfx) \mydef d_H(\bfx, \bfzero)$. An \emph{erasure} is a symbol of a word that one knows to be erroneous. Finally, we denote\footnote{take care that this notation (with $\le d$ instead of $<k$) is not the most currently used, but remains very convenient for our work} the full-length Reed-Solomon code by
\[
    \RS_q(d) \mydef \{ \ev_{\FF_q}(f), f \in \FF_q[T], \deg(f) \le d \}\,,
\]
and we recall that $\RS_q(d)$ can correct efficiently $1$ erasure and up to $\lfloor \frac{n-d}{2} \rfloor$ errors.

\begin{definition}[locally correctable code]
    Let $1 \le \ell \le k \le n$, and $\delta, \epsilon > 0$. A code $\calC \subseteq \FF_q^n$ is said $(\ell, \delta, \epsilon)$-locally correctable if there exists a probabilitic algorithm ${\rm Dec} : [1,n] \to \FF_q$ such that the following holds. For every $1 \le i \le n$ and for every $\bfy \in \FF_q^n$ such that $d_H(\bfy, \bfc) \le \delta n$ for some $\bfc \in \calC$, we have:
    \begin{itemize}
    \item the probability\footnote{taken over the internal randomness of the decoder ${\rm Dec}$} that ${\rm Dec}(i)$ outputs $c_i$ is larger than $1 - \epsilon$;
    \item ${\rm Dec}(i)$ reads at most $\ell$ coordinates of $\bfy$.
    \end{itemize}
\end{definition}

Similarly to the case of classical Reed-Muller codes and codes derived from those, weighted Reed-Muller codes can be locally corrected using their restrictions to ``lines''. For simplicity, we see a vector $\bfy \in \FF_q^{q^2}$ as a map $\FF_q^2 \to \FF_q$, using the bijection between $[1,q^2]$ and $\FF_q^2$ given by the evaluation map. Similarly, $\bfa \in \FF_q^q$ is seen as a map $\FF_q \to \FF_q$. One obtains the local correction procedure described in Algorithm~\ref{algo:decodage-local}.

\begin{algorithm}
    \KwData{A coordinate $\bfx = (x_1, x_2) \in \FF_q^2$ where to decode, and a oracle access to a word $\bfy : \FF_q^2 \to \FF_q$, where $\bfy = \bfc + \bfe$, $\bfc \in \calC$, and ${\rm wt}(\bfe) \le \delta q^2$.}
    \KwResult{The symbol $c_\bfx$, with high probability.}
    Pick at random an $\eta$-line $L \in \Phi_\eta$ such that $L(t_0) = \bfx$ for some $t_0 \in \FF_q$.\\
    Define $S = L(\FF_q)$ and $\bfz = \bfy_{|S} : \FF_q \mapsto \FF_q$.\\
    Consider $z_{t_0}$ as an erasure, and decode $\bfz$ in the Reed-Solomon code $\RS_q(d+1)$, giving a corrected codeword $\tilde{\bfz}$.\\
    Output the corrected value $\tilde{z}_{t_0}$.\\
    \caption{\label{algo:decodage-local} A local correction algorithm ${\rm Dec}$ for the weighted Reed-Muller code $\WRM_q^\eta(d)$.}
\end{algorithm}

According to Katz and Trevisan's terminology~\cite{KatzT00}, Algorithm~\ref{algo:decodage-local} is not \emph{perfectly smooth}, in the sense that the coordinate $y_\bfx$ is never queried. nevertheless, it can be made smooth following techniques described in~\cite[Chapter $2$]{JLthese}.

\begin{theorem}
    Let $\eta \ge 1$, $q$ be a prime power, and $\gamma \in (0,1)$ such that $q - \lfloor \gamma q \rfloor$ is even. For every $\delta \le \frac{1 - \gamma}{4}$, the weighted Reed-Muller code $\WRM^\eta_q( \lfloor \gamma q \rfloor)$ is $(q-1, \delta, \epsilon)$-locally correctable where $\epsilon \le \frac{2}{1 - \gamma}\delta$.
\end{theorem}

\begin{proof}
    Let $\bfy = \bfc + \bfe : \FF_q^2 \to \FF_q$ be a corrupted codeword, where $\bfc \in \WRM_q^\eta(d)$ and ${\rm wt}(\bfe) \le \delta q^2$. We define $E = \{ \bfx \in \FF_q^2 \mid e_\bfx \ne 0 \}$ the support of $\bfe$. The random variable representing the set of queries addressed by the local decoder is denoted by $A_\bfx$. It is clear that the algorithm succeeds if $|A_\bfx \cap E| \le w$, where $w = \frac{q -  \lfloor \gamma q \rfloor}{2} - 1$, since a Reed-Solomon of dimension $\lfloor \gamma q \rfloor+1$ can decode up to $1$ erasure and $w$ errors. Using Markov's inequality, the probability $p$ of success of Algorithm~\ref{algo:decodage-local} satisfies:
    \[
        p \ge 1 - \PP(|A_\bfx \cap E| \ge w + 1) \ge 1 - \frac{\mathbb{E}(|A_\bfx \cap E|)}{w+1}\,.
    \]
    Moreover, for every $\bfa \in \FF_q^2$, we have $\PP(\bfa \in A_\bfx) \le \frac{q-1}{q^2-1}$. Hence,
    \[
        \mathbb{E}(|A_\bfx \cap E|) = \sum_{\bfa \in E} \PP(\bfa \in A_\bfx) \le \delta q^2 \cdot \frac{q-1}{q^2-1} \le \delta q\,.
    \]
    Finally we get
    \[
        p \ge 1 - \frac{4 \delta q}{q - \lfloor \gamma q \rfloor} \ge 1 - \frac{2 \delta}{1 - \gamma}\,.
    \]
\end{proof}

\begin{remark}
    If $\eta \ge 2$, it is possible to get a sharper bound for the probability $p$ of success of Algorithm~\ref{algo:decodage-local}. Using Chebyshev's inequality (quite similarly to~\cite[Proposition 2.36]{JLthese}), one can indeed prove that $p \ge 1 - {\cal O}\big(\frac{\delta(1-\delta)}{q}\big)$. 
\end{remark}

\section{Application to private information retrieval}
\label{sec:PIR}

Private information retrieval (PIR) protocols are cryptographic protocols ensuring that a user can retrieve an entry $D_i$ of a remote database $D = (D_1, \dots, D_k)$, without revealing any information on the index $i \in [1, k]$ to the holder of the database. Additionally, it is also required that the communication cost (number of bits exchanged during the retrieval process) is sublinear in the size of the database.

Since its introduction by Chor, Goldreich, Kushilevitz and Sudan in 1995~\cite{ChorGKS95}, various kinds of PIR schemes have been designed according to the system constraints. In earliest PIR schemes, one assumes that the database is replicated over $\ell$ non-communicating honest-but-curious servers $S_1, \dots, S_\ell$. In this context the seminal result of Katz and Trevisan~\cite{KatzT00} --- which relates PIR protocols to the existence of so-called \emph{smooth locally decodable codes} --- induced many new constructions of PIR schemes, notably in~\cite{BeimelIKR02, Yekhanin08, Efremenko12, DvirG16}. These constructions eventually achieved $O(\exp(\sqrt{\log k \log \log k}))$ bits of communication for a $k$-entry database replicated on $\ell=2$ servers.

Motivated by the use of storage codes in distributed storage systems, a large amount of recent works focused on the case where the database is \emph{encoded} on the servers. In this context, entries of the database are usually very large (\emph{e.g.} movies), so that we can assume that the \emph{download} communication cost prevails over the upload one. Several works aimed at minimizing this cost depending on the storage system: Shah, Rashmi and Ramchandran~\cite{ShahRR14} considered the replication code as the storage code; Tajeddine, Gnilke and El Rouayheb~\cite{TajeddineGR18} MDS codes; Kumar, Rosnes and Graell i Amat~\cite{KumarRA17} arbitrary codes.

It is worth noticing that, following \emph{e.g.}~Beimel and Stahl~\cite{BeimelS02}, a few works also considered the more restrictive setting of colluding servers (\emph{i.e.} servers communicating with each other so as to collect information about the required item), byzantine servers (\emph{i.e.} servers able to produce wrong answers to user's queries) or unresponsive servers (servers unable to give ananswer to user's queries).

Finally, one should emphasise that families of PIR schemes referenced above mostly focus on decreasing the communication cost during the retrieval process. This is done at the expense of other crucial parameters, such as the computation complexity of the recovery, or the servers' storage overhead.

In this section, we show how the local properties of weighted Reed-Muller codes $\WRM^\eta_q(d)$ lead to very natural PIR protocols resisting to any set of $b$ byzantine, $u$ unresponsive and $t$ colluding servers --- provided that $2b + u + t \le q-d-1$ --- with moderate communication complexity but optimal computation complexity.

\subsection{Definitions}

\begin{definition}[private information retrieval]
    Let $D \in \FF_q^k$ be a remote database distributed on $\ell$ servers $S_1, \dots, S_\ell$, in such a way\footnote{Notice that we make no other assumption on the way (replication, encoding, etc.) the database is stored on the servers. We only require that the encoding map $D \mapsto (\bfc^{(1)}, \dots, \bfc^{(\ell)})$ is injective.} that we assume that each server $S_j$ stores a vector $\bfc^{(j)} \in \FF_q^m$.
    A \emph{private information retrieval (PIR)} protocol for $D$ is a tuple of algorithms $(\Query, \Answer, \Recover)$ such that:
    \begin{enumerate}
    \item $\Query$ is a probabilistic algorithm taking as input a coordinate $i \in [1,k]$, and providing a random tuple of \emph{queries} $\Query(i) = (q_1, \dots, q_\ell) \in \calQ^\ell$ for some finite set $\calQ$;
    \item $\Answer$ is a deterministic algorithm taking as input a server index $j \in [1, \ell]$, a query $q_j \in \calQ$ and the vector $\bfc^{(j)}$ stored by server $S_j$, and outputs an \emph{answer} $a_j \in \calA$, where $\calA$ is a finite set;
    \item $\Recover$ is a deterministic algorithm taking as input a coordinate $i \in [1, k]$, a tuple of queries $\bfq = (q_1, \dots, q_\ell) \in \calQ^\ell$ and a tuple of answers $\bfa = (a_1, \dots, a_\ell) \in \calA^\ell$, and which outputs a symbol $r \in \FF_q$ satisfying the following requirement. If $\bfq = \Query(i)$ and $\bfa = ( \Answer(j, q_j, \bfc^{(j)}) )_{1 \le j \le \ell}$, then:
    \begin{equation}
        \label{eq:def-PIR}
         D_i = \Recover (i, \bfq, \bfa)\,.
    \end{equation}
    \end{enumerate}
    We also say that a PIR protocol 
    \begin{itemize}
    \item[--] is \emph{$t$-private} (or resists to any \emph{collusion} of $t$ servers) if for every $T \subset [1,\ell]$, $|T| = t$, we have
    \[
        {\rm I}(\Query(i)_{|T}\;;\,i) = 0,
    \]
    where ${\rm I}(\cdot\;;\,\cdot)$ denotes the mutual information between random variables;
    \item[--] is \emph{robust against $b$ byzantine and $u$ unresponsive servers} if \eqref{eq:def-PIR} holds when up to $b$ symbols of $\bfa = (\Answer(j, q_j, \bfc^{(j)}))_{1 \le j \le \ell} \in \calA^\ell$ differ from the expected ones, and up to $u$ symbols of $\bfa$ are missing.
    \end{itemize}
\end{definition}

Let us now define some of the most studied parameters of PIR protocols.

\begin{definition}
    Let $(\Query, \Answer, \Recover)$ be a PIR protocol. We define:
    \begin{itemize}
    \item[--] its \emph{communication complexity} as $C_{\rm comm} \mydef \ell (\log(|\calQ|) + \log(|\calA|))$;
    \item[--] its  \emph{server computation complexity}, denoted $C^s_{\rm comp}$ as the maximal number of operations over $\FF_q$ necessary to compute $\Answer(j, q_j, \bfc^{(j)})$;
    \item[--] its \emph{storage rate} as the ratio $\frac{k}{\ell m}$.
    \end{itemize}
We finally say that a PIR protocol is \emph{computationally optimal for the servers} if $C^s_{\rm comp} \le 1$.
\end{definition}

\subsection{The PIR protocol}

We present in this section a PIR protocol based on weighted Reed-Muller codes. The protocol relies on a well-suited splitting of the encoded database over the servers, as it was originally done by Augot, Levy-dit-Vehel and Shikfa in~\cite{AugotLS14}

\begin{protocol}
    \label{protocol:PIR}
    Let $\calC = \WRM_q^\eta(d)$, and denote its dimension by $k$. Recall that a codeword $\bfc \in \calC$ can be seen as a map $\FF_q^2 \to \FF_q$. Let us also consider $q$ servers $(S_t)_{t \in \FF_q}$ indexed by elements of $\FF_q$.

    {\bf Initialisation.} The database $D \in \FF_q^k$ is encoded into a codeword $\bfc \in \calC$. For every $t \in \FF_q$, the server $S_t$ receives the part $\bfc_{|\{t\} \times \FF_q}$ of the codeword $\bfc$. Notice that $\bfc_{|\{t\} \times \FF_q}$ consists in $q$ symbols over $\FF_q$.
    
    {\bf Queries.} Assume one wants to retrieve $D_i$, for $1 \le i \le k$. One can always assume that the encoding map is systematic, hence $D_i = c_\bfx$ for some $\bfx = (x_1, x_2) \in \FF_q^2$. To define a vector of queries:
    \begin{itemize}
    \item Pick at random an $\eta$-line $L \in \Phi_\eta$ such that $L(t_0) = \bfx$ for some $t_0 \in \FF_q$.
    \item The server $S_{t_0}$ receives a random element $y_{t_0} \in \FF_q$.
    \item Server $S_t, t \ne t_0$ receives $y_t \in \FF_q$ such that $(t, y_t) = L(t)$.
    \end{itemize}
    
    {\bf Answers.} Upon receipt of $y_t \in \FF_q$, every server $S_t$ reads the entry $c_{(t, y_t)} \in \FF_q$ and sends it back to the user. 
    
    {\bf Recovery.} The user collects $\bfc' = ( c_{(t, y_t)} )_{t \in \FF_q}$ and runs an error-and-erasure correcting algorithm for $\RS_q(d)$ with input $\bfc'$. Then, the user returns the corrected symbol $c'_{(t_0, y_{t_0})}$.
\end{protocol}

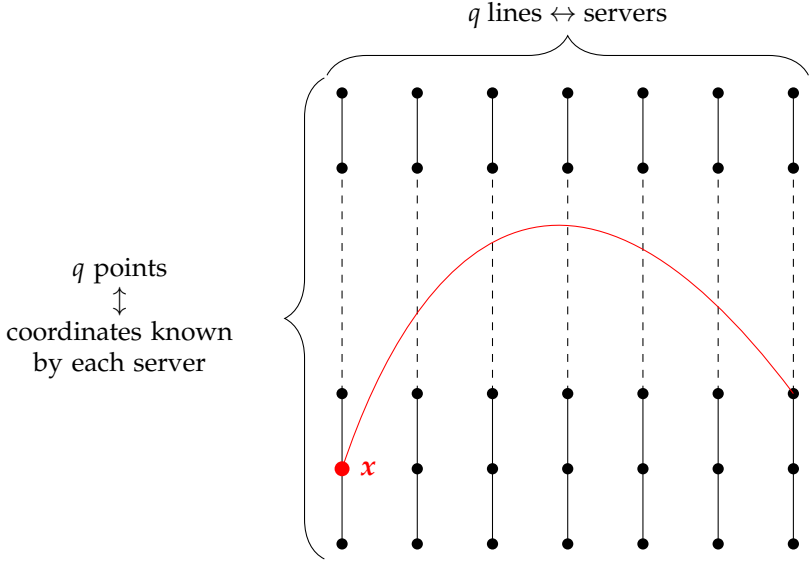
\begin{figure}[t]
\begin{tikzpicture}
\draw [-] (0,0) node[circle,inner sep=1.5pt,fill]{} -- (0,1)node[circle,inner sep=1.5pt,fill]{} -- (0,2)node[circle,inner sep=1.5pt,fill]{};
\draw [-,dashed] (0,2) -- (0,5)node[circle,inner sep=1.5pt,fill]{} ;
\draw [-] (0,5) -- (0,6)node[circle,inner sep=1.5pt,fill]{} ;

\draw [-] (1,0) node[circle,inner sep=1.5pt,fill]{} -- (1,1)node[circle,inner sep=1.5pt,fill]{} -- (1,2)node[circle,inner sep=1.5pt,fill]{};
\draw [-,dashed] (1,2) -- (1,5)node[circle,inner sep=1.5pt,fill]{} ;
\draw [-] (1,5) -- (1,6)node[circle,inner sep=1.5pt,fill]{} ;

\draw [-] (2,0) node[circle,inner sep=1.5pt,fill]{} -- (2,1)node[circle,inner sep=1.5pt,fill]{} -- (2,2)node[circle,inner sep=1.5pt,fill]{};
\draw [-,dashed] (2,2) -- (2,5)node[circle,inner sep=1.5pt,fill]{} ;
\draw [-] (2,5) -- (2,6)node[circle,inner sep=1.5pt,fill]{} ;

\draw [-] (3,0) node[circle,inner sep=1.5pt,fill]{} -- (3,1)node[circle,inner sep=1.5pt,fill]{} -- (3,2)node[circle,inner sep=1.5pt,fill]{};
\draw [-,dashed] (3,2) -- (3,5)node[circle,inner sep=1.5pt,fill]{} ;
\draw [-] (3,5) -- (3,6)node[circle,inner sep=1.5pt,fill]{} ;

\draw [-] (4,0) node[circle,inner sep=1.5pt,fill]{} -- (4,1)node[circle,inner sep=1.5pt,fill]{} -- (4,2)node[circle,inner sep=1.5pt,fill]{};
\draw [-,dashed] (4,2) -- (4,5)node[circle,inner sep=1.5pt,fill]{} ;
\draw [-] (4,5) -- (4,6)node[circle,inner sep=1.5pt,fill]{} ;

\draw [-] (5,0) node[circle,inner sep=1.5pt,fill]{} -- (5,1)node[circle,inner sep=1.5pt,fill]{} -- (5,2)node[circle,inner sep=1.5pt,fill]{};
\draw [-,dashed] (5,2) -- (5,5)node[circle,inner sep=1.5pt,fill]{} ;
\draw [-] (5,5) -- (5,6)node[circle,inner sep=1.5pt,fill]{} ;

\draw [-] (6,0) node[circle,inner sep=1.5pt,fill]{} -- (6,1)node[circle,inner sep=1.5pt,fill]{} -- (6,2)node[circle,inner sep=1.5pt,fill]{};
\draw [-,dashed] (6,2) -- (6,5)node[circle,inner sep=1.5pt,fill]{} ;
\draw [-] (6,5) -- (6,6)node[circle,inner sep=1.5pt,fill]{} ;

\draw [-,red] (0,1)node[circle,inner sep=2pt,red,fill, label=right:$\bfx$]{} .. controls (1,4) and (3,6) .. (6,2);

\draw [decorate, decoration={brace,amplitude=15pt,raise=1pt}] (-0.2,-0.2) -- (-0.2,6.2) 
            node [midway, xshift=-1mm, font=\small, outer sep=15pt, align=center, left, text width=4cm]{$q$ points \\
            $\updownarrow$\\
            coordinates known by each server};
            
\draw [decorate, decoration={brace,amplitude=15pt,raise=1pt}] (-0.2,6.2) -- (6.2,6.2) 
            node [midway, font=\small, outer sep=15pt, above]{$q$ lines $\leftrightarrow$ servers};            

\end{tikzpicture}
\caption{Illustration of the retrieval process. For a desired coordinate $c_\bfx$, an $\eta$-line $L$ (in red) containing $\bfx$ is picked at random.}\label{pir2}
\end{figure}

\begin{theorem}
    Let $q$ be a prime power, $\eta \ge 1$ , and $b, u \ge 0$. Set $d = q - u - 2b - 2$. Then, Protocol~\ref{protocol:PIR} equipped with $\WRM^\eta_q(d)$ is $\eta$-private and robust against $b$ byzantine and $u$ unresponsive servers. Moreover, it is computationally optimal for the servers, its storage rate approaches $1/{2\eta}$ when $q \to \infty$, and its communication complexity is $2 q \log q$.
\end{theorem}

\begin{proof}
    The correctness of the PIR scheme, under $b$ byzantine and $u$ unresponsive servers, comes from  Proposition~\ref{prop:restriction} and from the fact that $\RS_q(d)$ corrects $b$ errors and $u+1$ erasures if $d \ge q - u - 2b - 2$. Moreover, the scheme is $\eta$-private since any subset of $\eta$ points of an $\eta$-line gives no information about the other points. Finally, the parameters of the scheme can be easily checked.
\end{proof}

\section{Towards higher information rate: the lifting process}
\label{sec:lift}

\subsection{Definitions}

In previous sections, we have proved that weighted Reed-Muller codes admit local properties that can be used in practical applications such as private information retrieval. However, such constuctions are moderately efficient in terms of storage, since the information rate of $\WRM^\eta_q(d)$ is bounded by $1/2\eta$ if $d \le q-2$.

In this section, we show how to construct codes with the same local properties as weighted Reed-Muller codes, but admitting a much larger dimension. As a practical consequence, these new codes can replace weighted Reed-Muller codes in Protocol~\ref{protocol:PIR}, leading to storage-efficient PIR schemes.

Techniques involved in the construction of these codes directly follow the lifting process initiated by Guo, Kopparty and Sudan~\cite{GuoKS13}. More precisely, the authors introduce so-called \emph{lifted Reed-Solomon codes} as codes containing (classical) Reed-Muller codes, and satisfying that the restriction of any codeword to any affine line lies in a Reed-Solomon. The purpose of this section is to extend this notion to $\eta$-lines.

We thus naturally introduce the $\eta$-lifting of a Reed-Solomon code as follows.

\begin{definition}[$\eta$-lifting of a Reed-Solomon code]
    Let $q$ be a prime power and $0 \le d \le q-1$. The \emph{$\eta$-lifting of the Reed-Solomon code $\RS_q(d)$} is the code of length $n = q^2$ defined as follows:
    \[
        \Lift^\eta (\RS_q(d)) \mydef \{ \ev_{\FF_q^2}(f) \mid f \in \FF_q[X,Y], \forall L \in \Phi_\eta, \ev_{\FF_q}(f \circ L) \in \RS_q(d) \}\,. 
    \]
\end{definition}

Notice that if $d=q-1$, the $\eta$-lifted code $\Lift^\eta (\RS_q(q-1))$ is the trivial full space $\FF_q^{q^2}$. Hence, from now on we assume $d \le q-2$.

It is clear that $\WRM^\eta_q(d) \subseteq \Lift^\eta (\RS_q(d))$ since the constraints that define $\eta$-lifted codes are satisfied by each codeword of a comparable weighted Reed-Muller code. But quite surprisingly, the code $\Lift^\eta (\RS_q(d))$ is sometimes much larger than $\WRM^\eta_q(d)$. Let us highlight this claim with an example.

\begin{example}
    Let $q=4$, $\eta=2$ and $d=2$. The associated weighted Reed-Muller code is generated by the evaluation vectors of monomials $X^iY^j$, where $(i,j)$ lies in
    \[
    \{ (0,0), (0,1), (1,0), (2,0) \}\,.
    \]
    Let us now consider the monomial $f(X,Y) = Y^2 \in \FF_4[X,Y]$ and an $\eta$-line $L(T) = (T, a T^2 + b T + c) \in \Phi_2$, where $a, b, c \in \FF_4$.  We see that for every $t \in \FF_4$, we have:
    \[
    (f \circ L)(t) = (a t^2 + b t + c)^2 = a^2 t^4 + b^2 t^2 + c^2 = b^2 t^2 + a^2 t + c\,.
    \]
    Hence, $\ev_{\FF_4}(f \circ L) \in \RS_4(2)$ for every $L \in \Phi_2$. Since $\wdeg_\eta(f) = 4 > 2$, we get
    \[
        \ev_{\FF_4^2}(f) \in \Lift^2(\RS_4(2)) \setminus \WRM^2_4(2)\,.
    \]
\end{example}

Given a polynomial $f(X,Y) = \sum_{i,j} f_{i,j} X^iY^j \in \FF_q[X,Y]$, we define its \emph{degree set} as 
\[
\Deg(f) \mydef \{ (i,j) \in \NN^2, f_{i,j} \ne 0 \}\,.
\]
By extension, the degree set $\Deg(S)$ of a subset $S \subseteq \FF_q[X,Y]$ is the union of degree sets of polynomials lying in $S$. Similarly, if $\calC = \{ \ev_{\FF_q^2}(f), f \in S \}$, then we set $\Deg(\calC) \mydef \Deg(S)$. 

\begin{remark}
    Since $a^q = a$ for every $a \in \FF_q$, one can consider degree sets as subsets of $[0, q-1]^2$. This precisely corresponds to considering polynomials modulo the ideal $I = \langle X^q-X, Y^q-Y \rangle = \ker \ev_{\FF_q^2}$.
\end{remark}


\begin{lemma}
    \label{lem:aux1}
    Let $f \in \FF_q[X,Y]$ such that $\Deg(f) \subseteq [0,q-1]^2$, and let $(i, j) \in \Deg(f)$. Assume that for every $(a,b) \in \Deg(f)$, we have $i \ge a$ (respectively, $j \ge b$). Then, there exists an $\eta$-line $L \in \Phi_\eta$ such that $\deg(f \circ L) = i$ (respectively, $\deg(f \circ L) = j$). 
\end{lemma}

\begin{proof}
    If $i \ge a$ for every $(a,b) \in \Deg(f)$, then $L(T) = (T,1)$ lies in $\Phi_\eta$, and the degree of $f \circ L$ is thus $i$. The proof is similar for $j$.
\end{proof}

\begin{proposition}
    \label{prop:degree-set-square}
    Let $d \le q-2$. Then,
    \[
        \Deg(\Lift^\eta(\RS_q(d))) \subseteq [0,d]^2\,. 
    \]
\end{proposition}

\begin{proof}
    A pair $(i,j) \in \Deg(\Lift^\eta(\RS_q(d))) \setminus [0,d]^2$ would contradict Lemma~\ref{lem:aux1}. 
\end{proof}

\subsection{Monomiality}

We say that a linear code $\calC$ is \emph{monomial} if there exists a set $S \subset \FF_q[X,Y]$ of monomials, such that $\calC = \Span \{ \ev_{\FF_q^2}(f), f \in S \}$. Monomial codes are convenient since they admit a simple description.

Let us define monomial transformations $m_{a,b}: (x,y) \mapsto (ax,by)$, for $(a,b) \in (\FF_q^\times)^2$.

\begin{lemma}
    \label{lem:invariance-gives-monomiality}
    Let $S$ be a subspace of $\FF_q[X,Y]$ such that:
    \begin{enumerate}[label=(\roman*)]
    \item $\Deg(S) \subseteq [0,q-2]^2$, and
    \item for every $f(X,Y) \in S$ and every $(a,b) \in (\FF_q^\times)^2$, the polynomial $f \circ m_{a,b}$ also lies in $S$.
    \end{enumerate}
    Then $S$ is spanned by monomials.
\end{lemma}

\begin{proof}
Let $f(X,Y) = \sum_{(i,j) \in D} f_{i,j} X^iY^j \in S$ where $D = \Deg(f) \subseteq [0,q-2]^2$. It is sufficient to prove that for all $(i,j) \in D$,  the monomial $X^iY^j$ lies in $S$.

For $(i,j) \in D$, let us define
\[
Q_{i,j}(X, Y) \mydef \sum_{(a,b) \in (\FF_q^\times)^2} \frac{1}{a^ib^j} f(aX,bY)\,.
\]
Since $S$ is a vector space invariant under $\{m_{a,b} \mid (a,b) \in (\FF_q^\times)^2\}$, we have $Q_{i,j} \in S$. Moreover,
\[
  \begin{aligned}
    Q_{i,j}(X,Y) &= \sum_{(a,b) \in (\FF_q^\times)^2} \frac{1}{a^ib^j} \left( \sum_{(d,e) \in \Deg(f)} f_{d,e} \,a^d b^e X^dY^e \right)\\
    &= \sum_{(d,e) \in \Deg(f)} f_{d,e} \sum_{(a,b) \in (\FF_q^\times)^2} a^{d-i} b^{e-j} X^dY^e \\
    &= \sum_{(d,e) \in \Deg(f)} f_{d,e}  \cdot \underbrace{\Big( \sum_{a \in \FF_q^\times} a^{d-i}\Big)}_{=0 \text{ if } d=i,\, -1 \text{ otherwise}} \cdot \underbrace{\Big( \sum_{b \in \FF_q^\times} b^{e-j} \Big)}_{=0 \text{ if } e=j,\, -1 \text{ otherwise}}  \cdot \; X^dY^e\\
    &= f_{i,j} \cdot (-1)^2 \cdot X^iY^j\,.
  \end{aligned}
  \]
  Since $f_{i,j} \neq 0$, $X^iY^j \in S$.
\end{proof}

\begin{proposition}
    Let $d \le q-1$. The linear code $\Lift^\eta(\RS_q(d))$ is monomial.
\end{proposition}
\begin{proof}
    The code $\Lift^\eta(\RS_q(q-1))$ is the full space $\FF_q^{q^2}$; hence it is trivially a monomial code. For $d \le q-2$, let us define
    \[
        S \mydef \{ f \in \FF_q[X,Y], \Deg(f) \subseteq [0,q-1]^2, \ev_{\FF_q^2}(f) \in \Lift^\eta(\RS_q(d)) \}\,.
        \]
    Proposition~\ref{prop:degree-set-square} ensures that $\Deg(S) \subseteq [0,d]^2$. Let $f = \sum_{i,j} f_{i,j} X^i Y^j \in S$. For every $(a, b) \in (\FF_q^\times)^2$ and every $L(T) = (T, \phi(T)) \in \Phi_\eta$ we have
    \[
        f \circ m_{a,b} \circ L(T) = \sum_{i,j} f_{i,j} a^i T^i b^j \phi(T)^j\,.
    \]
    Let us now define $Q(T) \mydef f(T, b\phi(a^{-1}T))$. One can easily check that $(T, b\phi(a^{-1}T))) \in \Phi_\eta$. Since $\ev_{\FF_q^2}(f) \in \Lift^\eta(\RS_q(d))$, we also know that $\ev_{\FF_q}(Q) \in \RS_q(d)$. Moreover, $\RS_q(d)$ is invariant under affine transformations, hence $\ev_{\FF_q}(Q(aT)) \in \RS_q(d)$. Let us now remark that
    \[
        Q(aT) = \sum_{i,j} f_{i,j} a^i T^i b^j \phi(T)^j = f\circ m_{a,b} \circ L(T)\,.
    \]
    Consequently, $f\circ m_{a,b}\in S$. Therefore we can use Lemma~\ref{lem:invariance-gives-monomiality}, and our result follows immediately.
\end{proof}

\subsection{The degree set of $\eta$-lifted Reed-Solomon codes}

Previous discussions ensure that, given a tuple $(\eta, d, q)$, the code $\calC(q, d, \eta) \mydef \Lift^\eta(\RS_q(d))$ is fully determined by its \emph{degree set} $D(q, d, \eta) \mydef \Deg(\calC(q, d, \eta)) \subseteq [0,d]^2$. Let us now seek for characterisations of $D(q, d, \eta)$.

For this purpose, we need to introduce some notation:
\begin{itemize}[label=--]
    \item $\ps{\cdot}{\cdot}$ denotes the inner product between vectors, or tuples.
    \item We set $\bfw \mydef (1, 2, \dots, \eta) \in \NN^\eta$.
    \item Given $\alpha \in \NN$ and a prime number $p$, we denote by $\alpha^{(r)}$ the $r^{\rm th}$ digit in the representation of $\alpha$ in base $p$, \textit{i.e.} $\alpha=\sum_{r\geq 0} \alpha^{(r)} p^r$.
    \item For $\alpha, \beta \in \NN$, we write $\alpha \le_p \beta$ if and only if $\alpha^{(r)} \le \beta^{(r)}$ for every $r \ge 0$.  
    \item For $\bfk \in \NN^\eta$ and $r \in \NN$, we also write $\bfk^{(r)}=\big(k_1^{(r)},\dots,k_\eta^{(r)}\big) \in \NN^\eta$.
\end{itemize}

We will also make use of Lucas theorem~\cite{Lucas78} which gives the reduction of binomial coefficients modulo primes.
\begin{theorem}[Lucas theorem~\cite{Lucas78}]
    Let $a, b \in \NN$ and $p$ be a prime number. Recall that $a = \sum_{i \ge 0} a^{(i)} p^i$ is the representation of $a$ in base $p$. Then,
    \[
        \binom{a}{b} = \prod_{i \ge 0} \binom{a^{(i)}}{b^{(i)}} \: \mod\!p \,.   
    \]
\end{theorem}
In particular, in any field of characteristic $p$, the binomial coefficient $\binom{a}{b}$ is non-zero if and only if $b \le_p a$.


In the next lemma, we characterise univariate polynomials arising from the restriction of $Y^j$ to $\eta$-lines.
\begin{lemma}
    \label{lem:restriction-polynomial}
    Let $j \ge 0$ and $\eta \ge 1$ and let us define $\Phi_\eta^j \mydef \{ \phi(T)^j \mid \phi(T) \in \FF_q[T], \deg \phi \le \eta \} \subseteq \FF_q[T]$. We have:
    \[
        \Phi_\eta^j = \Span \{ T^\alpha \mid \alpha \in \Delta(j, \eta) \}\,,
    \]
    where 
    \[
        \Delta(j, \eta) \mydef \Big\{ \ps{\bfw}{\bfk} \mid \bfk \in \NN^\eta
        \text{ such that } \forall m \le \eta, \: k_m \le_p j - \sum_{\ell = 1}^{m-1} k_\ell \Big\}\,.
    \]
\end{lemma}

\begin{proof}
   Given a polynomial $\phi(T) = \sum_{m=0}^\eta a_m T^m \in \FF_q[T]$, the well-known multinomial theorem entails that:
    \[
    \begin{aligned}
        \phi(T)^j
            &= (a_0 + a_1 T + \dots + a_\eta T^\eta)^j\\
            &= \sum_{k_1 + \dots + k_\eta \le j} \binom{j}{k_1, \dots, k_\eta} \lambda_\bfk x^{k_1 + 2 k_2 + \dots + \eta k_\eta},
    \end{aligned}           
    \]
    where $\lambda_\bfk \mydef a_0^{j-|\bfk|} \times \prod_{\ell = 1}^\eta a_\ell^{k_\ell} \in \FF_q$ is a coefficient which only depends on $a_0, \dots, a_\eta$ and $\bfk$, and where
    \[
    \binom{j}{\bfk} \mydef \binom{j}{k_1, \dots, k_\eta} = \frac{j!}{k_1! k_2! \dots k_\eta! (j- \sum_{m=1}^\eta k_m)!}\,.
    \]
    
    The coefficient of the term $T^\alpha$ in $\phi(T)^j$ is therefore:
    \[
    c_\alpha = \sum_{\bfk \in K_\alpha} \binom{j}{\bfk} \lambda_\bfk\,, 
    \]
    where $K_\alpha \mydef \{ \bfk \in \NN^\eta \mid | \bfk| \le j \text{ and } \ps{\bfw}{\bfk} = \alpha \}$.  We claim that $c_\alpha= 0$ for every $\phi \in \Phi_\eta$ if and only if $\binom{j}{\bfk} = 0$ for every $\bfk \in K_\alpha$. Indeed, $c_\alpha \in \FF_q$ can be seen as the evaluation of an homogeneous polynomial $C_\alpha \in \FF_q[A_0, \dots, A_\eta]$ of degree $j$ at the point $(a_0, \dots, a_\eta) \in \FF_q^{\eta+1}$ corresponding to $\phi$. Since $j \le q-1$, the polynomial $C_\alpha$ vanishes over $\FF_q^{\eta+1}$ if and only if it is the zero polynomial, which proves our claim.
    
    Now, notice that
    \[
    \binom{j}{\bfk} = \binom{j}{k_1}\binom{j-k_1}{k_2}\binom{j-k_1-k_2}{k_3}  \cdots \binom{j-k_1-k_2-\dots-k_{\eta-1}}{k_\eta}\,.
    \]
    Hence, using Lucas theorem~\cite{Lucas78} on every binomial coefficient in the above product, we see that $\binom{j}{\bfk} = 0$ if and only if there exists $m \in [1, \eta]$ such that $k_m \not\le_p j - \sum_{\ell=1}^{m-1} k_\ell$.
    
    In other words, the monomial $T^\alpha$ appears as a term of $\phi(T)^j$ if and only if there exists $\bfk \in \NN^\eta$ such that $\alpha = \ps{\bfw}{\bfk} = \sum_{\ell = 1}^\eta \ell k_\ell$ and 
    \[
        \forall m \in [1, \eta], k_m \le_p j - \sum_{\ell=1}^{m-1} k_\ell\,.
    \]
\end{proof}

Let us now give some properties on the set $\Delta(j, \eta) \subseteq \NN$ defined in Lemma~\ref{lem:restriction-polynomial}.

\begin{lemma}
    \label{lemDj}
    We have $\Delta(j, \eta) \subseteq [0, j \eta]$. Moreover, an integer $\alpha$ belongs to $\Delta(j, \eta)$ if and only if 
    \begin{equation}
        \label{eq:lem-cara-Dj}
        \exists \bfk \in \NN^\eta \text{ such that } \alpha= \ps{\bfw}{\bfk} \text{ and } \forall r \geq 0, \: \sum_{\ell=1}^m k^{(r)} \leq j^{(r)}.
    \end{equation}
\end{lemma}

\begin{proof}
By definition, an integer $\alpha$ belongs to $\Delta(j, \eta)$ if and only if there exists $\bfk \in \NN^\eta$ such that $\alpha = \sum_{\ell=1}^\eta \ell k_\ell$ and for all $m \leq \eta$, we have  
\begin{equation}
    \label{eq:defDj}
    k_m \le_p j - \sum_{\ell=1}^{m-1} k_\ell\,.
\end{equation}

We first prove by induction on $m$ that, if $\alpha \in \Delta(j,\eta)$, then for all $m \leq \eta$ and for all $r \geq 0$, 
\[
    \sum_{\ell=1}^m k_\ell^{(r)} \leq j^{(r)}\,.
\]
Notice that it would prove the desired result for $m=\eta$. Moreover, the case $m=1$ is a direct consequence of \eqref{eq:defDj}.

Let us fix $2 \leq m \leq \eta$ such that $\sum_{\ell=1}^{m-1} k_\ell^{(r)} \leq j^{(r)}$ for every $r \ge 0$. Then $\sum_{\ell=1}^{m-1} k_\ell^{(r)} \leq p-1$ and the uniqueness of the representation of the integer $\sum_{\ell=1}^{m-1} k_\ell$ in base $p$ ensures that
\begin{equation}
    \label{eq:demo-cara-Delta}
    \left(\sum_{\ell=1}^{m-1} k_\ell\right)^{(r)}= \sum_{\ell=1}^{m-1} k_\ell^{(r)} \leq j^{(r)}.
\end{equation}
Using \eqref{eq:defDj}, we get $k_m^{(r)} \leq j^{(r)} - \sum_{\ell=1}^{m-1} k_\ell^{(r)}$, which implies that $\sum_{\ell=1}^{m} k_\ell^{(r)}\leq j^{(r)}$.

Conversely, assume that \eqref{eq:lem-cara-Dj} holds, and let $1 \le m \le \eta$. We shall prove that \eqref{eq:defDj} is satisfied. For every $r\geq 0$, we have
\[
    k_m^{(r)} \le \sum_{\ell=m}^\eta k_\ell^{(r)} = \sum_{\ell=1}^\eta k_\ell^{(r)} - \sum_{\ell=1}^{m-1} k_\ell^{(r)}.
\]
Equation \eqref{eq:lem-cara-Dj} implies that $ k_m^{(r)} \le j^{(r)} - \sum_{\ell=1}^{m-1} k_\ell^{(r)}$. Moreover, $\sum_{\ell=1}^{m-1} k_\ell^{(r)} \le j^{(r)}$, hence as we have seen in \eqref{eq:demo-cara-Delta},
\[
    \left(\sum_{\ell=1}^{m-1} k_\ell\right)^{(r)} = \sum_{\ell=1}^{m-1} k_\ell^{(r)}\,.
\]
This leads us to $k_m^{(r)} \le \left(j - \sum_{\ell=1}^{m-1} k_\ell\right)^{(r)}$. Therefore, $k_m \le_p j - \sum_{\ell=1}^{m-1} k_\ell$.
\end{proof}



As an easy corollary of Lemma~\ref{lem:restriction-polynomial} and Lemma~\ref{lemDj}, we see that
\[
    \Deg( \{ (X^iY^j) \circ \phi, \phi \in \Phi_\eta \} ) = \{ i + u, u \in \Delta(j, \eta) \}\,.
\]
Hence, $\ev_{\FF_q^2}(X^iY^j)$ lies in $\Lift^\eta \RS_q(d)$ if, for all $u \in \Delta(j, \eta)$, every monomial $T^{i + u}$ evaluates to a codeword of $\RS_q(d)$. Notice here that $i+u$ might be larger than $q$, therefore this is equivalent to say that $T^{i+u} \mod (T^q - T)$ is polynomial of degree bounded by $d$.

This remark leads us to introduce a relation of equivalence between integers. We write $a \equivstar{q} b$ if and only if $T^a = T^b \mod (T^q- T)$. In other words, $a \equivstar{q} b$ if and only if $(a,b) = (0, 0)$, or $a > 0, b > 0$ and $(q-1) \mid (a - b)$. Finally, we denote\footnote{notation $\mod\!^* q$ is used in~\cite{GuoKS13}, but we find it quite unconvenient} by $\Red{q}(a)$ the only integer in $[0, q-1]$ such that $\Red{q}(a) \equivstar{q} a$.

From Lemma~\ref{lem:restriction-polynomial} and Lemma~\ref{lemDj}, and following the previous discussion, we deduce a characterisation of elements of $D(q, d, \eta)$.

\begin{proposition}
    \label{caracLRS}
    Let $d \le q-2$. A pair $(i,j) \in [0,d]^2$ belongs to $D(q, d, \eta)$ if and only if for every $\bfk \in \NN^{\eta}$ such that for all $r \ge 0, |\bfk^{(r)}| \leq j^{(r)}$, we have
    \[
        \Red{q}(i + \ps{\bfw}{\bfk}) \le d.
    \]
\end{proposition}

\section{Analyses of sequences of degree sets}
\label{sec:degree-set}

For a generic tuple $(\eta, q, d)$, it seems difficult to give an explicit description of the degree set of $\Lift^\eta \RS_q(d)$. Our approach is to analyse \emph{sequences} of degree sets $D(q, d, \eta)$ with varying parameters $q = p^e$, $d$, and $\eta$, in order to produce good asymptotic families of codes.

We will illustrate our analyses with graphical representations of degree sets. Our convention is the following. Assume one wants to represent a degree set $D \subseteq [q-1]^2$. If $(i,j) \in D$, then a black (or sometimes grey) unit square is represented at coordinate $(i, j)$; otherwise, a white unit square is plotted. Such an illustration is proposed in Example~\ref{ex:degree-set}.

\begin{example}
    \label{ex:degree-set}
    The degree set $D$ of $\Lift^2 (\RS_8(5))$, namely
    \[
    D = \{ (0,0), (1,0), (2,0), (3,0), (4,0), (5,0), (1,0), (1,1), (1,2), (1,3), (2,0), (2,1), (4,0), (4,1), (4,4) \}
    \]
    is represented in Figure~\ref{fig:ex-p2-r3-d5-eta2}.
    
\begin{figure}[!ht]
    \centering
    \includegraphics[scale=0.35]{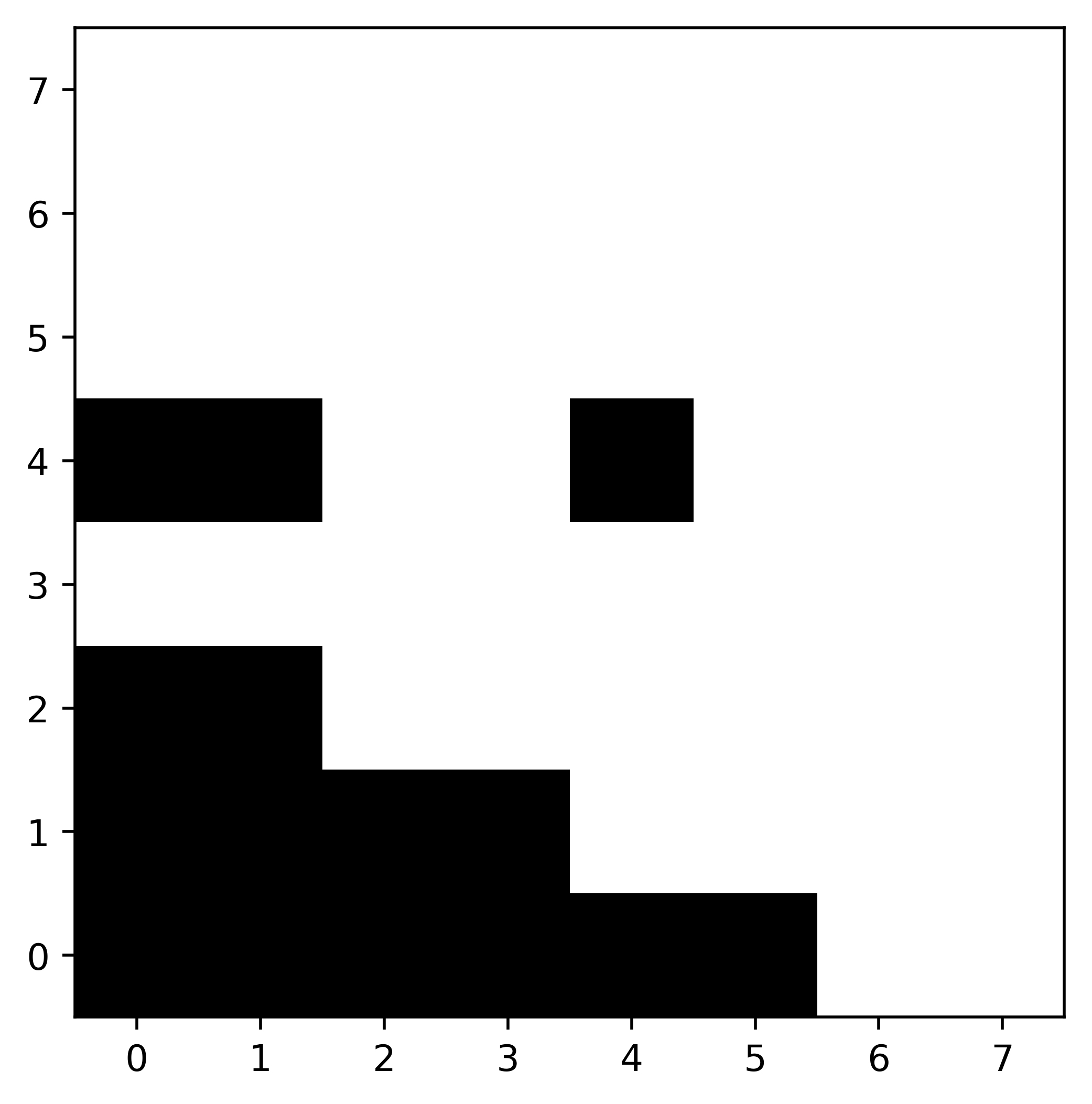}
    \caption{A representation of the degree set $D$ of $\Lift^2 \RS_8(5)$. 
    }
    \label{fig:ex-p2-r3-d5-eta2}
\end{figure}
\end{example}

Let us now provide generic relations between $\eta$-lifted codes of varying parameters.

\subsection{Increasing and decreasing sequences of $\eta$-lifted codes}

\subsubsection{Sequence $(D(q, d, \eta))_{\eta \ge 1}$, with $(q,d)$ fixed and varying $\eta$}

\begin{lemma}
Let us fix a prime power $q$ and $d \le q-1$. The sequence of codes $(\Lift^\eta \RS_q(d))_{\eta \ge 1}$ is decreasing with respect to the inclusion of codes.
\end{lemma}

\begin{proof}
It is enough to notice that an $\eta$-line is also an $(\eta+1)$-line, therefore every codeword of $\Lift^{\eta+1} \RS_q(d)$ fulfills the constraints defining $\Lift^\eta \RS_q(d)$.
\end{proof}

In Figure~\ref{fig:incresing-eta}, we plot a sequence of degree sets which illustrates this result on $\FF_{16}$.

\begin{figure}[!ht]
\centering
\begin{subfigure}[t]{0.3\textwidth}
\centering
\includegraphics[width=0.9\textwidth]{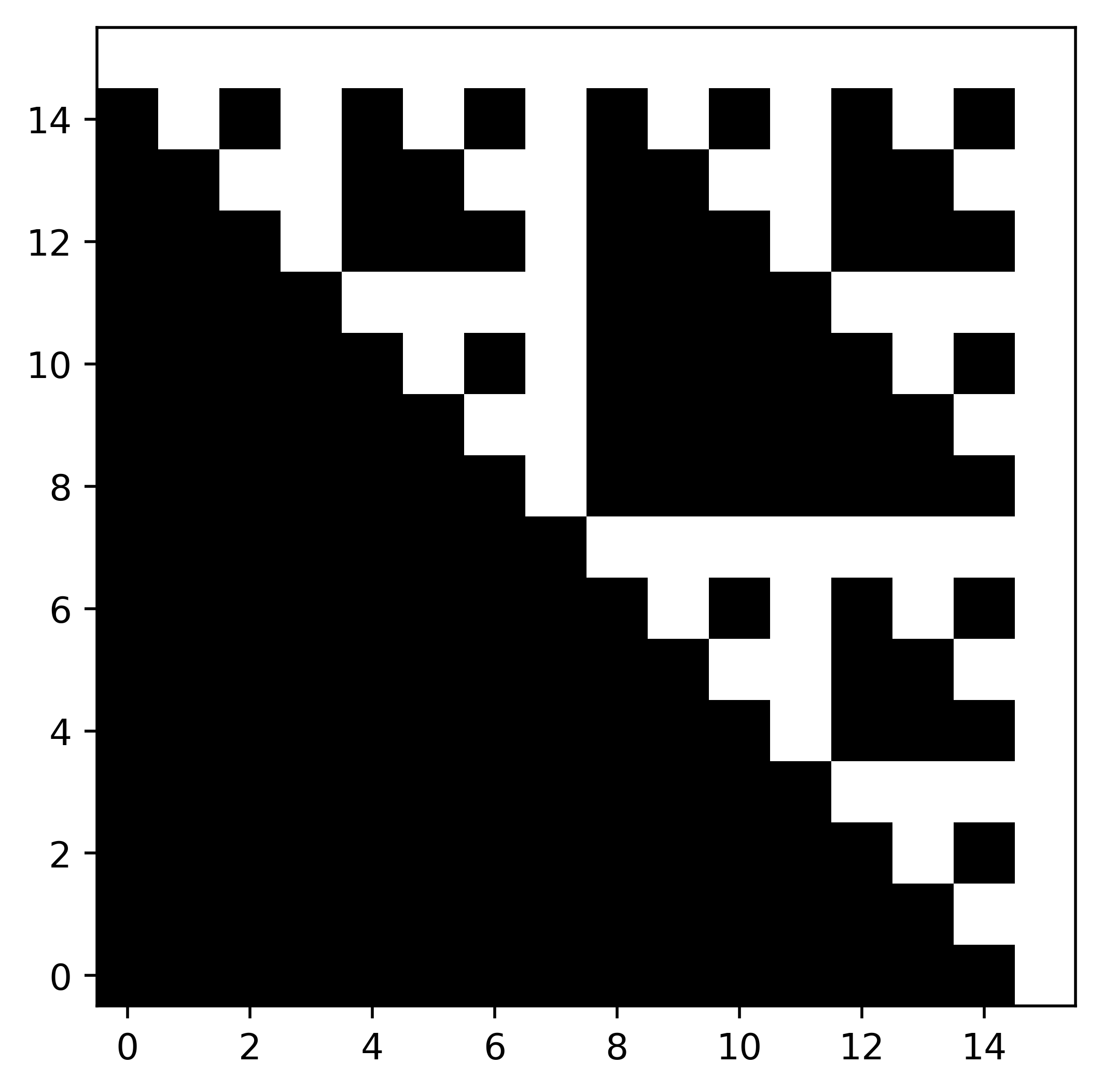}
\subcaption{$\eta=1$}
\end{subfigure}%
\begin{subfigure}[t]{0.3\textwidth}
\centering
\includegraphics[width=0.9\textwidth]{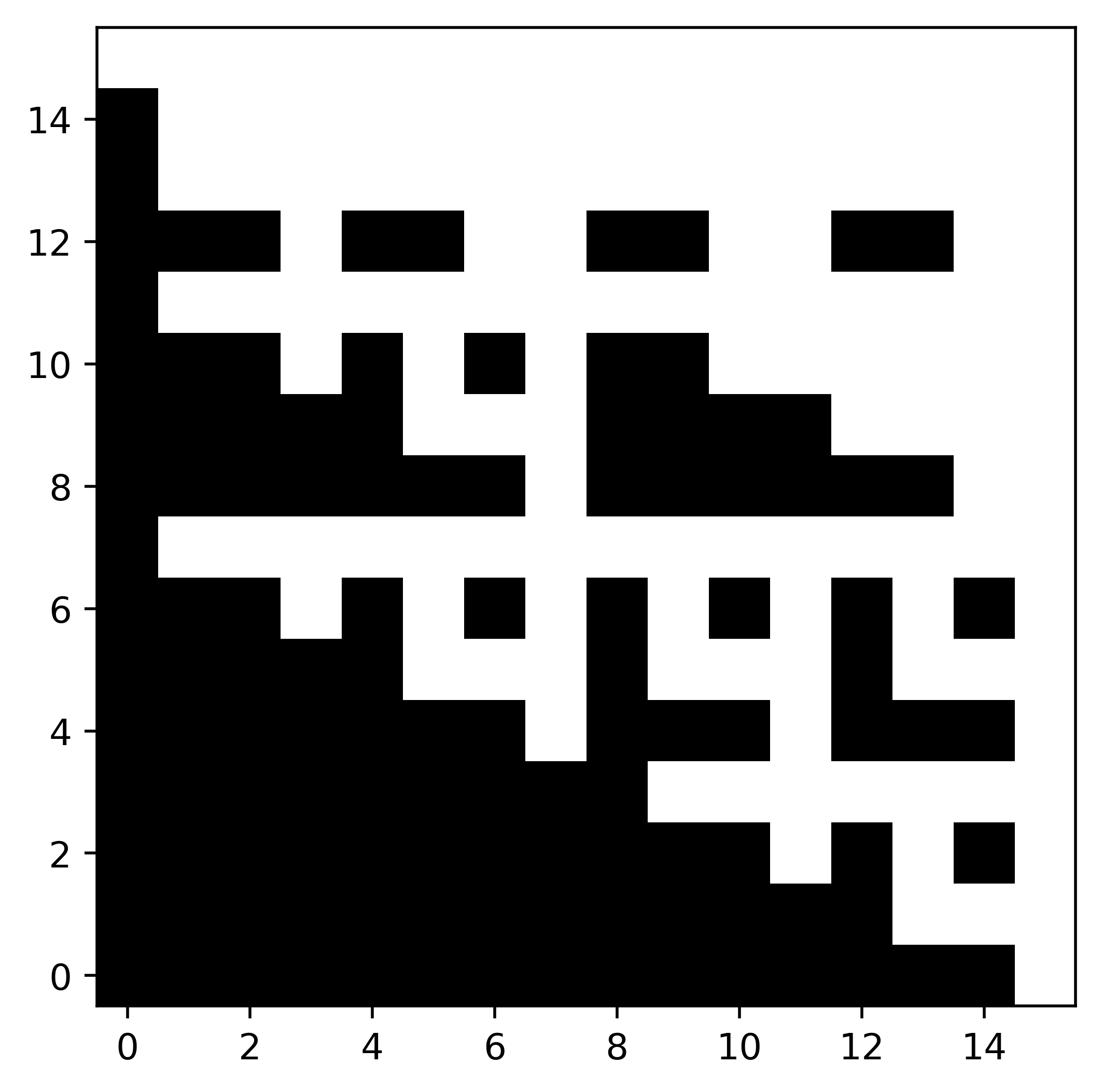}
\subcaption{$\eta=2$}
\end{subfigure}%
\begin{subfigure}[t]{0.3\textwidth}
\centering
\includegraphics[width=0.9\textwidth]{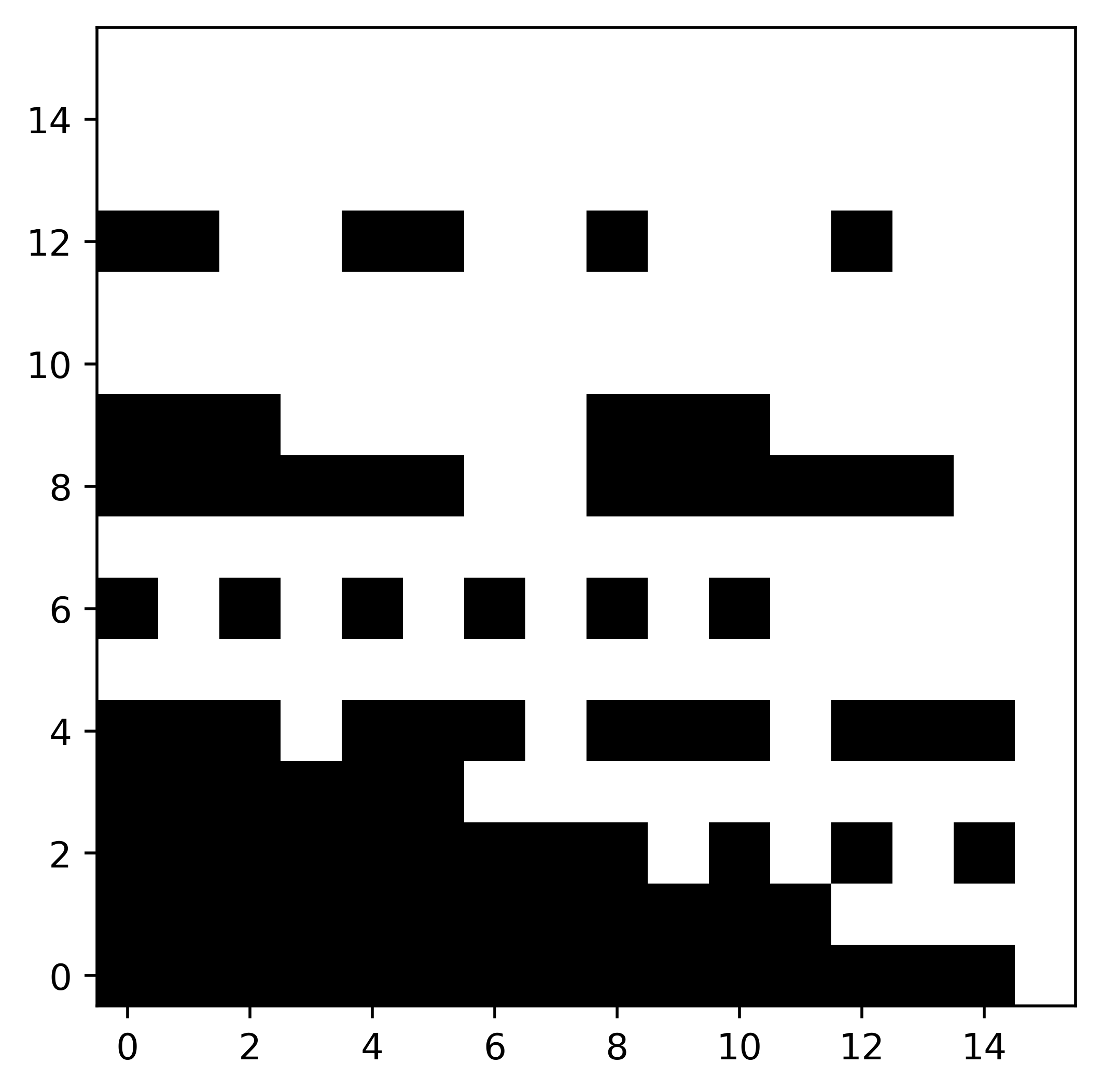}
\subcaption{$\eta=3$}
\end{subfigure}

\begin{subfigure}[t]{0.3\textwidth}
\centering
\includegraphics[width=0.9\textwidth]{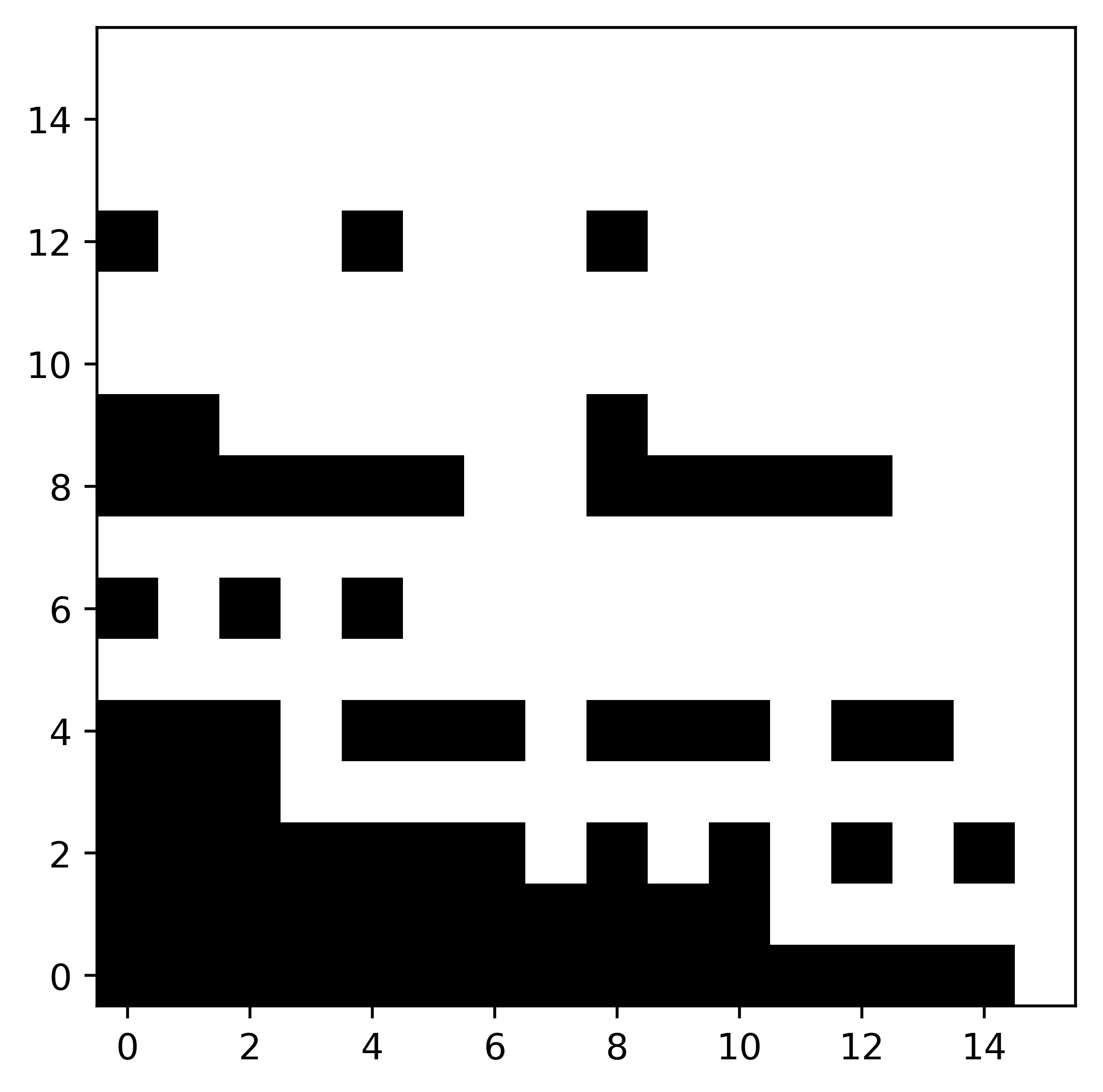}
\subcaption{$\eta=4$}
\end{subfigure}%
\begin{subfigure}[t]{0.3\textwidth}
\centering
\includegraphics[width=0.9\textwidth]{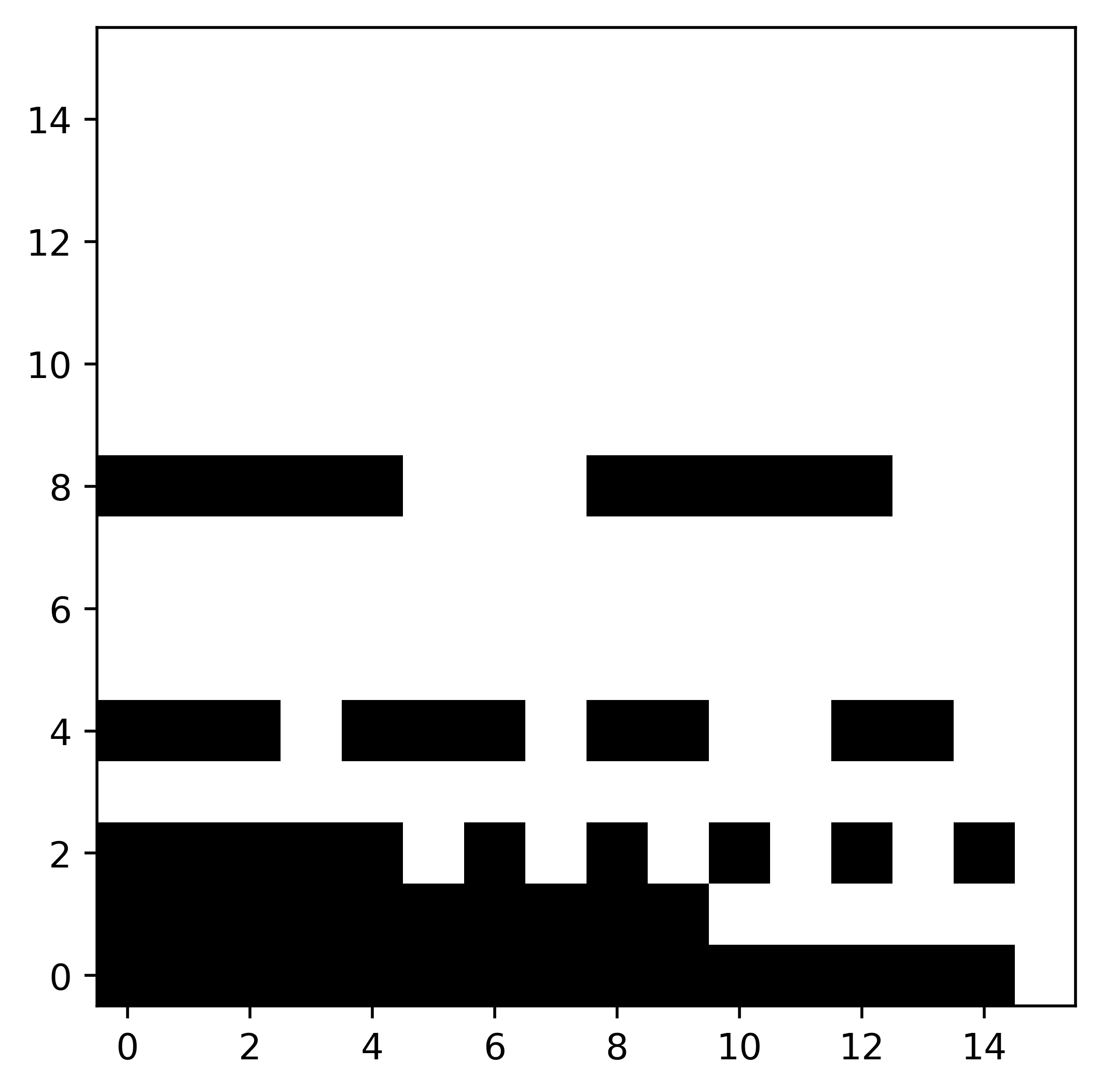}
\subcaption{$\eta=5$}
\end{subfigure}%
\begin{subfigure}[t]{0.3\textwidth}
\centering
\includegraphics[width=0.9\textwidth]{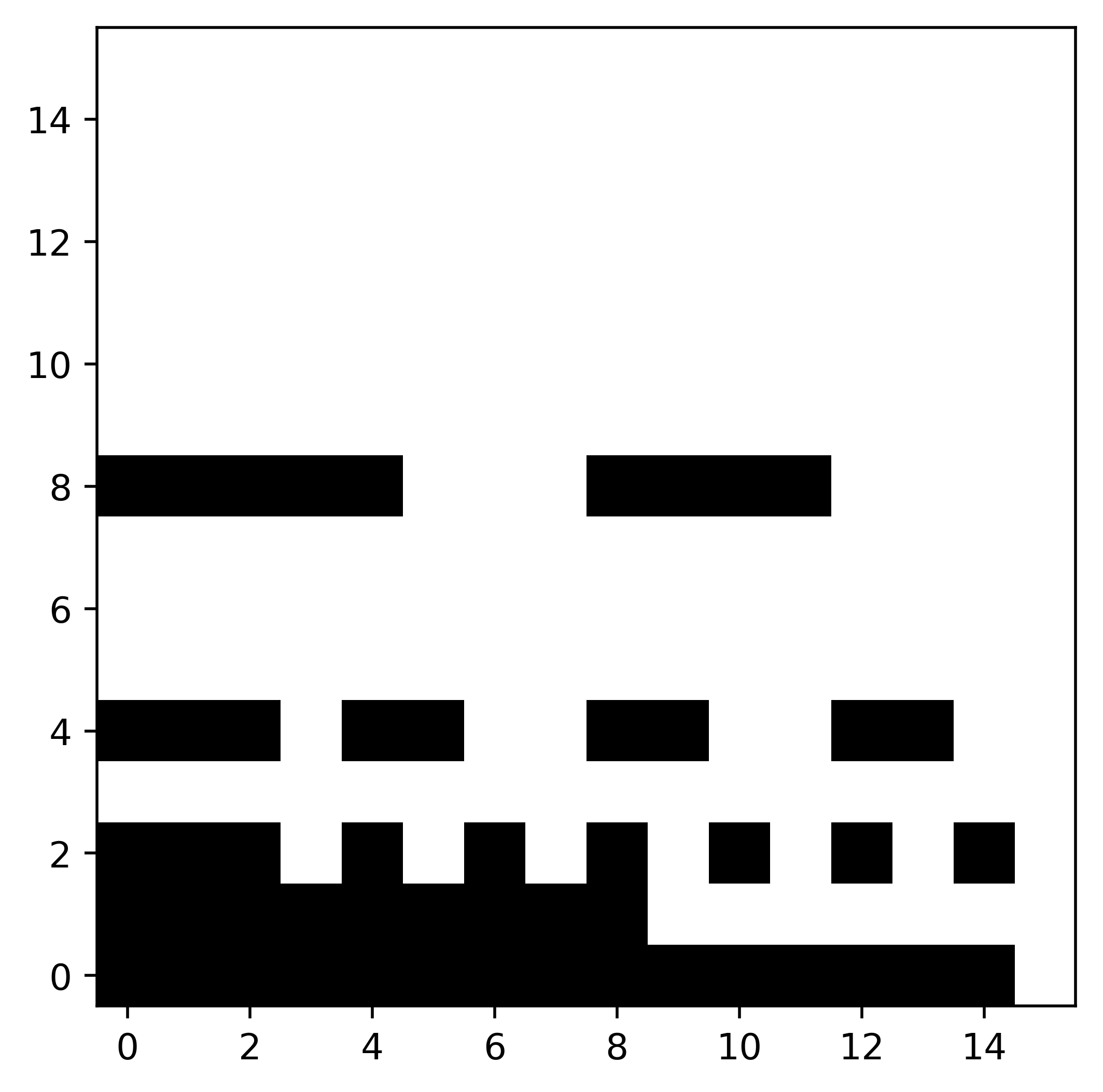}
\subcaption{$\eta=6$}
\end{subfigure}
\caption{\label{fig:incresing-eta}Representation of the degree set of $\Lift^\eta \RS_{16}(14)$ for different values of $\eta$}
\end{figure}

\subsubsection{Sequence $(D(q, d, \eta))_{0 \le d \le q-2}$ with $(q,\eta)$ fixed and varying $d$}

\begin{lemma}
    \label{dcroiss}
    Let us fix a prime power $q$ and $\eta \ge 1$. The sequence $(\Lift^\eta \RS_q(d))_{d \ge 0}$ is increasing.
\end{lemma}

\begin{proof}
It is a straightforward consequence of the embedding of $\RS_q(d)$ into $\RS_q(d+1)$.
\end{proof}

In Figure~\ref{fig:incresing-degree}, we plot a sequence of degree sets which illutrates this result on $\FF_{16}$ with $\eta=2$.

\begin{figure}[ht]
\centering
\begin{subfigure}[t]{0.3\textwidth}
\centering
\includegraphics[width=0.9\textwidth]{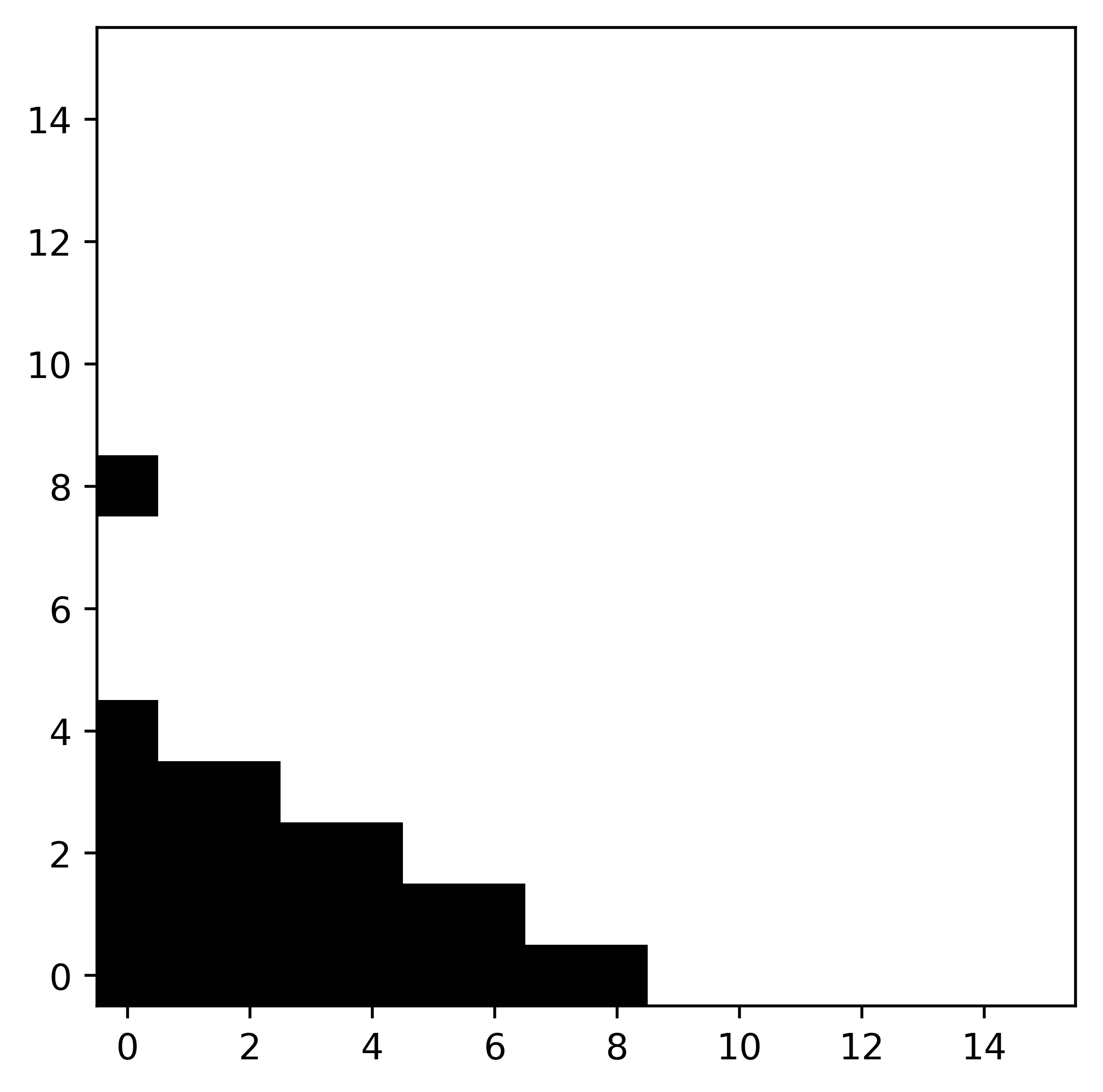}
\subcaption{$d=8$}
\end{subfigure}%
\begin{subfigure}[t]{0.3\textwidth}
\centering
\includegraphics[width=0.9\textwidth]{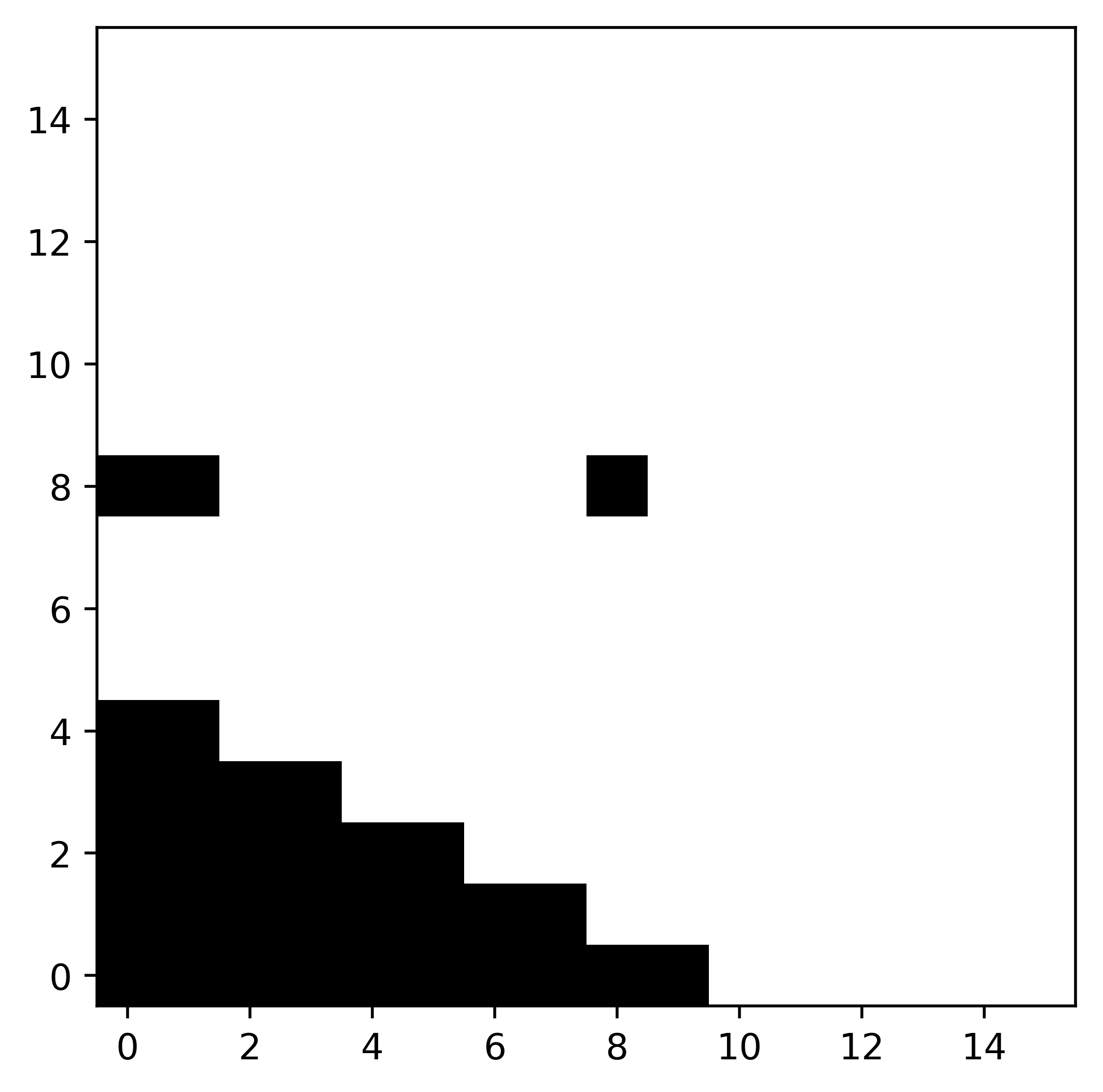}
\subcaption{$d=9$}
\end{subfigure}%
\begin{subfigure}[t]{0.3\textwidth}
\centering
\includegraphics[width=0.9\textwidth]{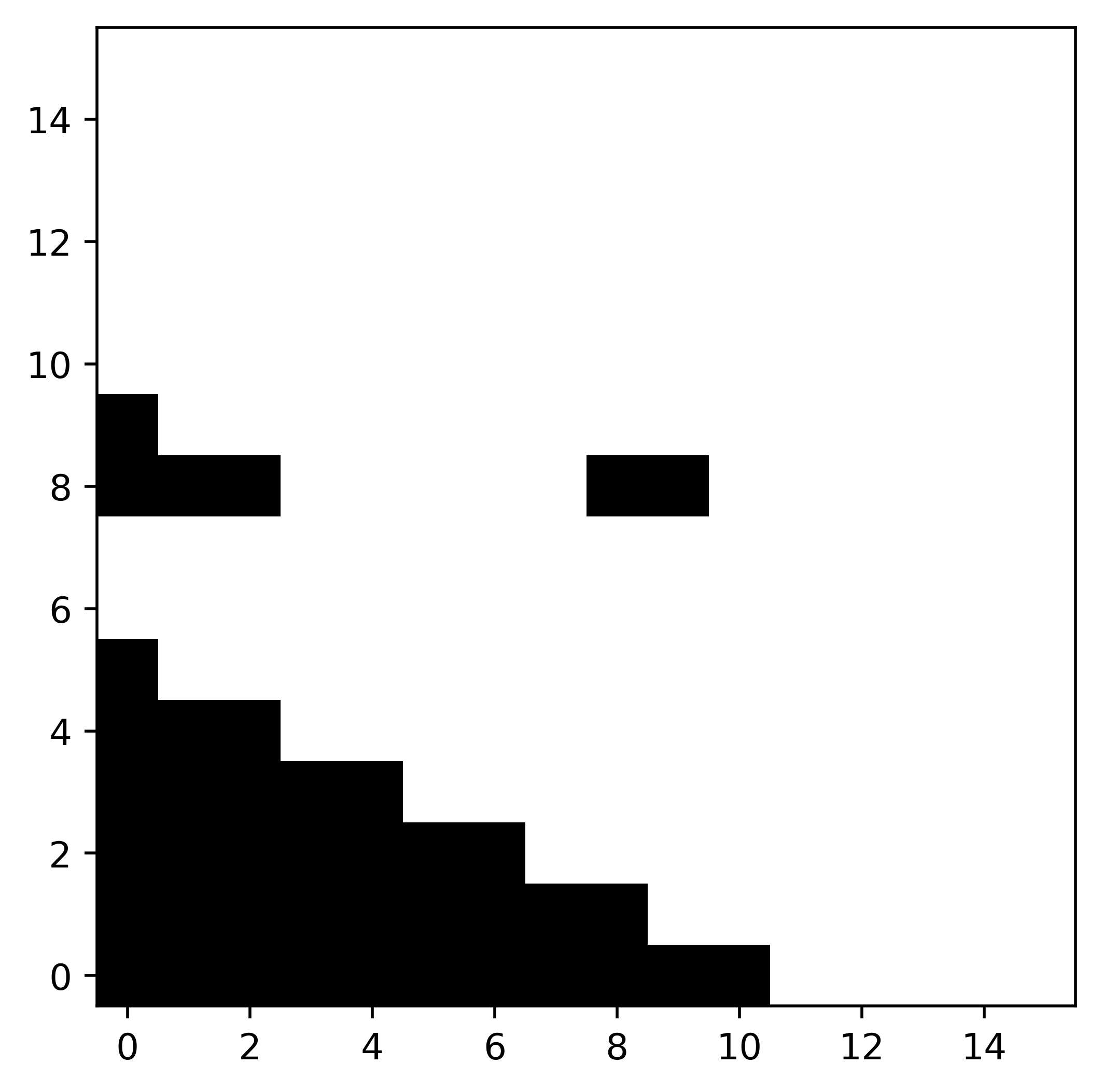}
\subcaption{$d=10$}
\end{subfigure}

\begin{subfigure}[t]{0.3\textwidth}
\centering
\includegraphics[width=0.9\textwidth]{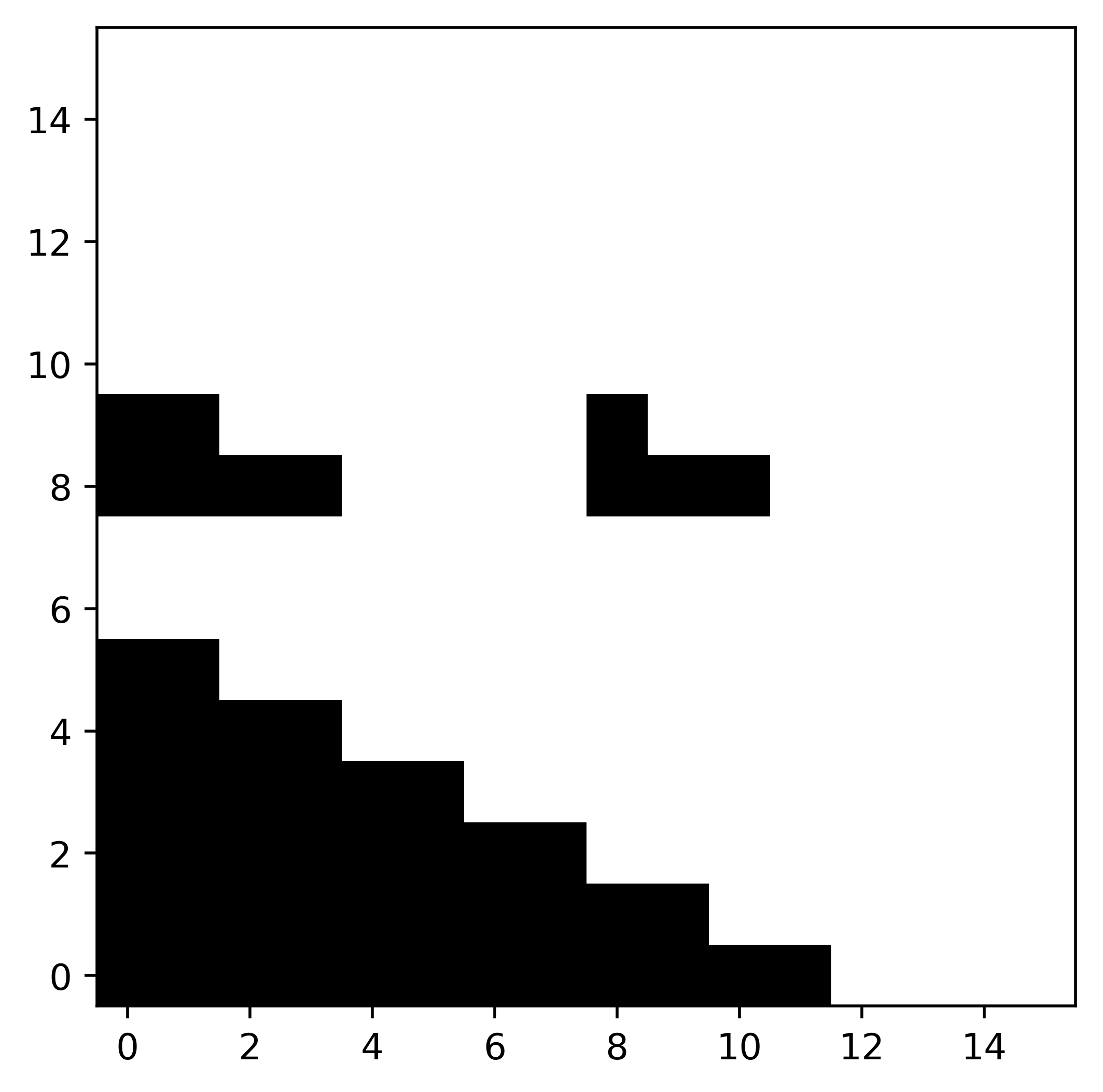}
\subcaption{$d=11$}
\end{subfigure}%
\begin{subfigure}[t]{0.3\textwidth}
\centering
\includegraphics[width=0.9\textwidth]{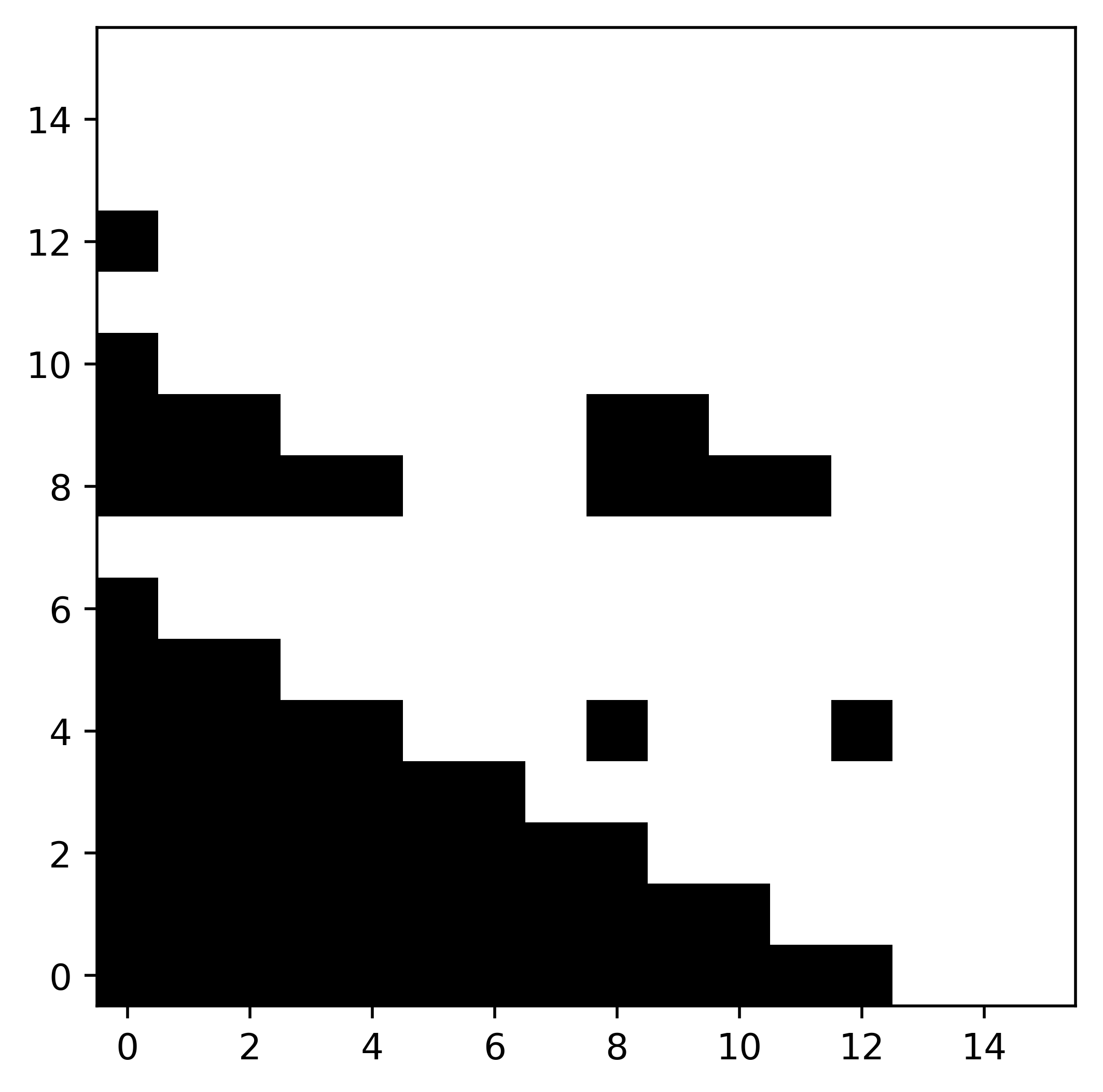}
\subcaption{$d=12$}
\end{subfigure}%
\begin{subfigure}[t]{0.3\textwidth}
\centering
\includegraphics[width=0.9\textwidth]{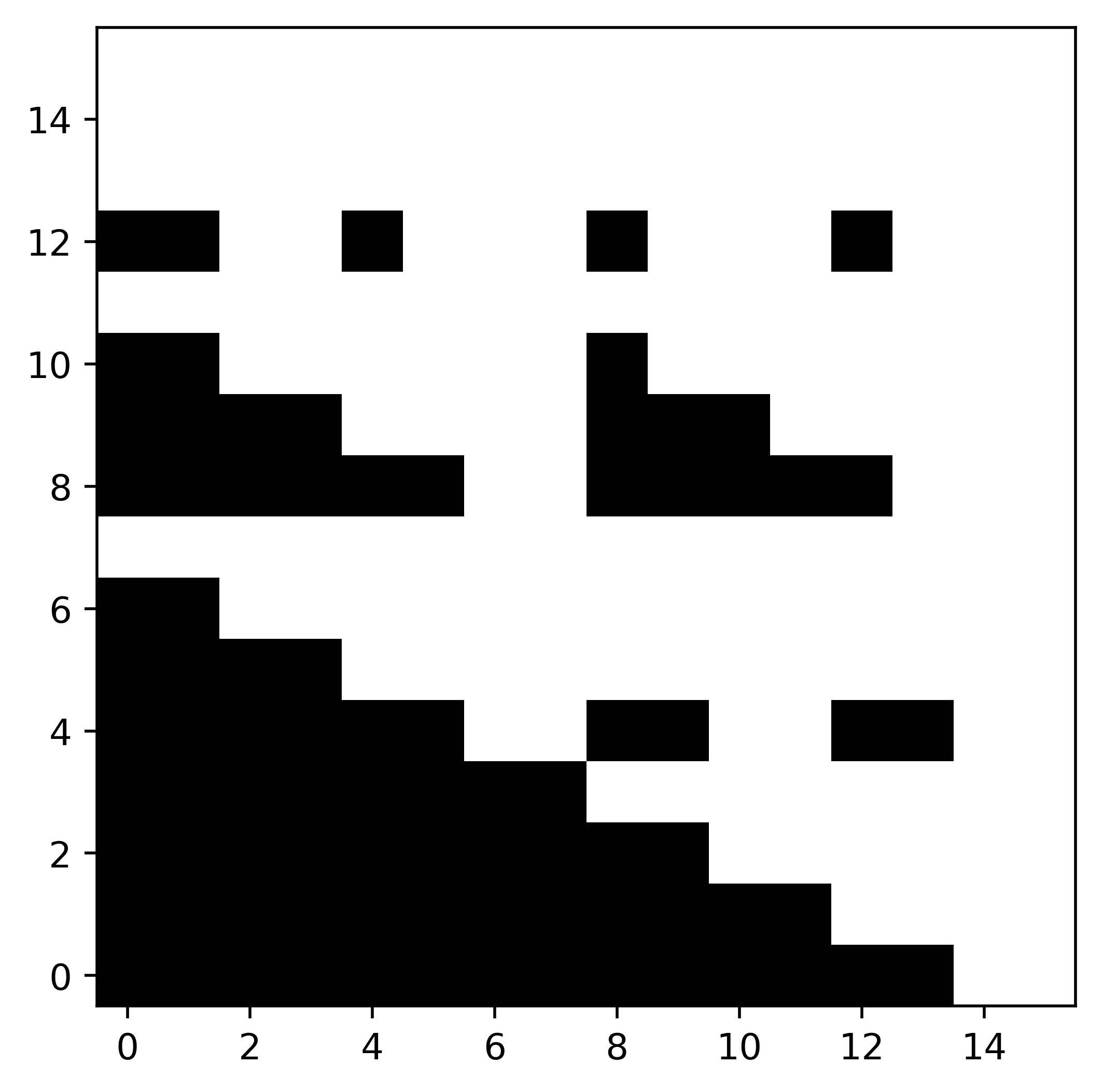}
\subcaption{$d=13$}
\end{subfigure}
\caption{Representation of the degree set of $\Lift^2 \RS_{16}(d)$ for different values of $d$}\label{fig:incresing-degree}
\end{figure}

\subsubsection{Sequence $( D(q, q - \alpha,\eta) )_{q}$ with fixed $(\alpha, \eta)$, and varying $q$}

Let us fix a prime number $p$, and let us consider a sequence of degree sets $( D(p^e, p^e - \alpha, \eta) )_{e \ge 1}$ with fixed $(\alpha, \eta)$, and varying $e$. Figure~\ref{fig:inclusion-degree-sets} represents such a sequence. In this figure, one can notice that $D(p^e, p^e - \alpha, \eta)$ is a subpattern (highlighted in grey) of the larger degree sets $D(p^{e+1}, p^{e+1} - \alpha, \eta)$.

\begin{figure}[ht]
\centering
\begin{subfigure}[t]{0.3\textwidth}
\centering
\includegraphics[width=0.9\textwidth]{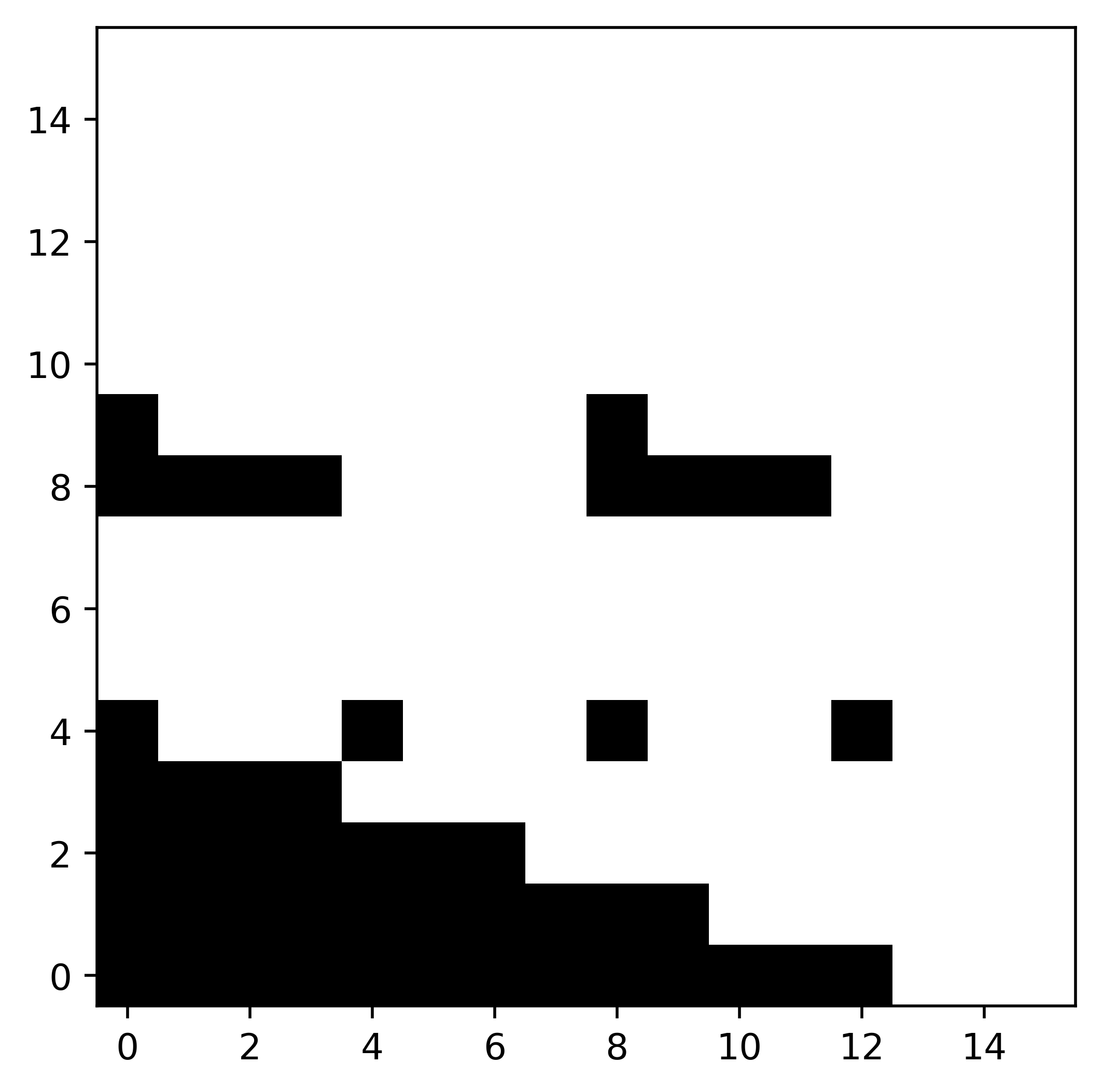}
\subcaption{$e=4$}
\end{subfigure}%
\begin{subfigure}[t]{0.3\textwidth}
\centering
\includegraphics[width=0.9\textwidth]{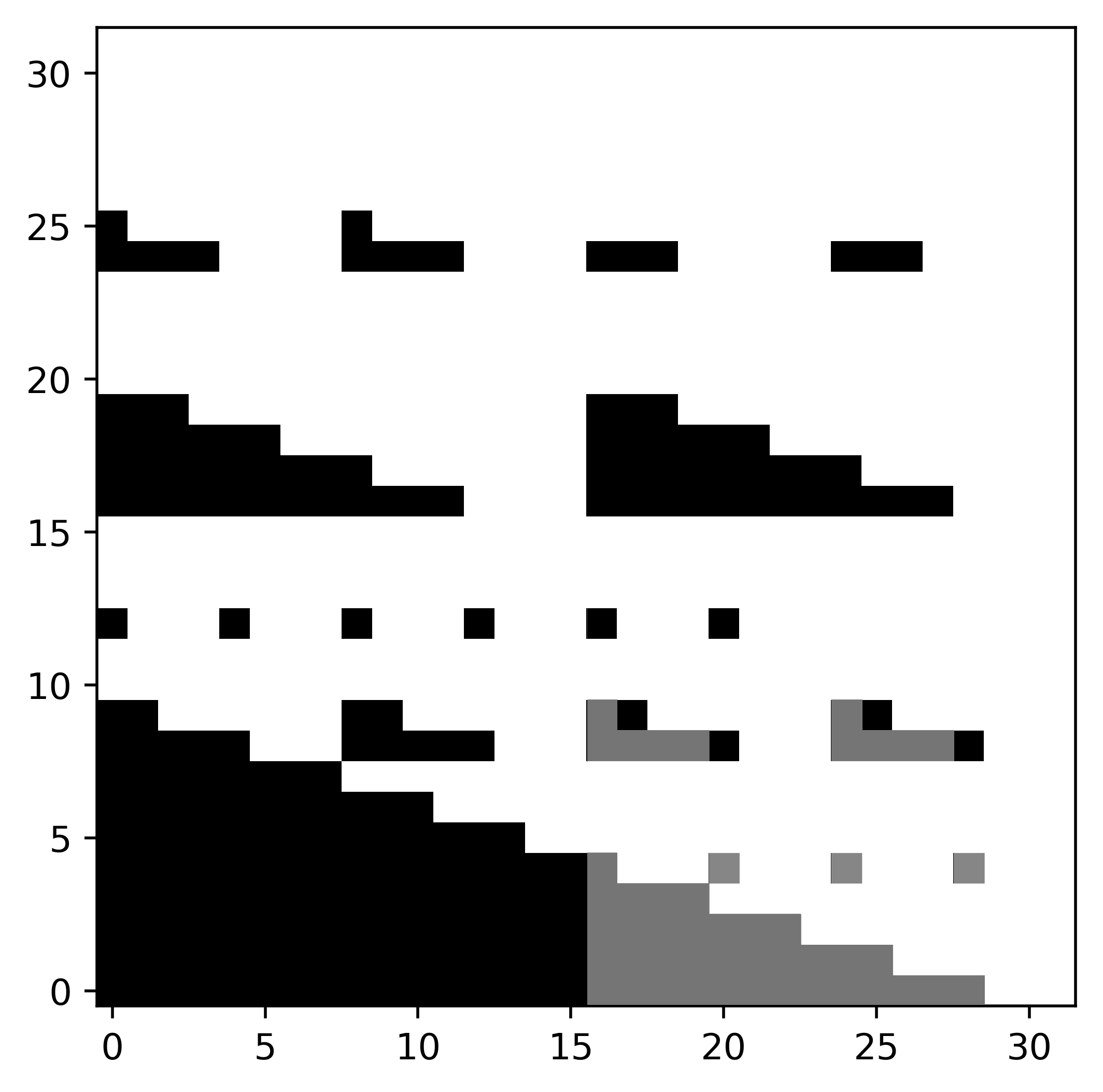}
\subcaption{$e=5$}
\end{subfigure}%
\begin{subfigure}[t]{0.3\textwidth}
\centering
\includegraphics[width=0.9\textwidth]{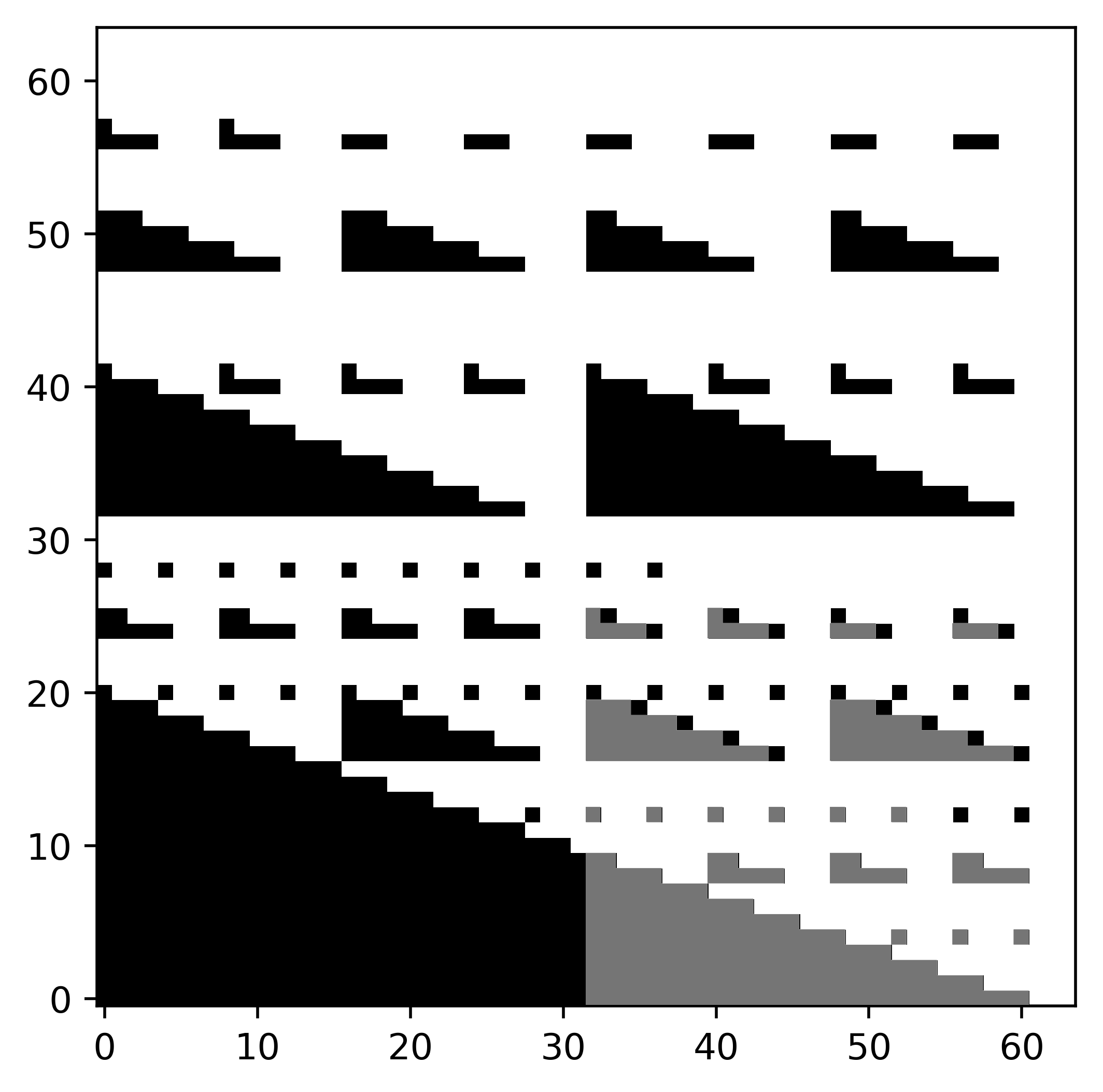}
\subcaption{$e=6$}
\end{subfigure}
\caption{
    \label{fig:inclusion-degree-sets}
    Representation of the degree set of $\Lift^3 \RS_{2^e}(2^e-4)$ for increasing values of $e$. In the degree set over $\FF_{2^e}$, the grey part is an exact copy of the degree set over $\FF_{2^{e-1}}$ which is represented on its left.
    }
\end{figure}

This remark seems trivial at first, but it has a meaningful consequence in terms of codes. Indeed, it shows that the corresponding $\eta$-lifted codes are (up to isomophism) subcodes to each other when the field size $q = p^e$ grows. This property is formalized in the following lemma.

\begin{lemma}
    \label{lem: seq d}
    Let $\eta < q = p^e$ and $2 \le \alpha \le p^e$. If $(p^e - i,j) \in D(p^e, p^e - \alpha, \eta)$, then
    \[
         (p^{e+1} - i, j) \in D(p^{e+1}, p^{e+1} - \alpha, \eta)\,.
    \]  
\end{lemma}

\begin{proof}
    Let $(p^e - i, j) \in D(p^e, p^e - \alpha, \eta)$, and consider $\bfk \in \NN$ such that $|\bfk^{(r)}| \leq j^{(r)}$ for every $r \geq 0$. Using Proposition~\ref{caracLRS}, we know that $\Red{p^e}((p^e - i) + \ps{\bfw}{\bfk}) \le p^e - \alpha$, and we want to prove that $\Red{p^{e+1}}(p^{e+1} - i) \le p^{e+1} - \alpha$.
    
    Notice that there exists $(Q_0, Q_1, R) \in \NN^3$ satisfying:
    \[
        (p^e - i) + \ps{\bfw}{\bfk} = (Q_1 p + Q_0)(p^e - 1) + R
    \]
    with $Q_0 \leq p-1$ and $R \leq p^e - \alpha$.  Since $\ps{\bfw}{\bfk} \le \eta |\bfk| \le \eta j \le \eta(p^e - 1)$, one can also check that $Q_1 p + Q_0 \le \eta + 1$.
    
    The case $R = 0$ must be handled at first. Notice that this implies that $(p^e - i) + \ps{\bfw}{\bfk} = 0$, meaning that $(p^e-i, j) = (0,0)$. Then one can check that $(p^{e+1} - p^e, 0) \in D(p^{e+1}, p^{e+1} - \alpha, \eta)$ since $\alpha \le p^e$. Hence, from now on, we assume that $R \ge 1$, and we distinguish two cases.
    
    
   First, assume that $Q_0 \geq 1$. Then we have
    \[
        p^{e+1} - i + \ps{\bfw}{\bfk}
            = p^{e+1} - p^e + (Q_1 p + Q_0)(p^e - 1) + R
            = (Q_1 + 1) (p^{e+1} - 1) + R'
    \]
    where
    \[
        R' \mydef  Q_0(p^e - 1) + R - (Q_1 + 1)(p - 1)\,.
    \]
    We see that $p^{e+1} - i + \ps{\bfw}{\bfk} \equivstar{p^{e+1}} R'$, hence it is sufficient to prove that $1 \le R' \le p^{e+1} - \alpha$.  Using $R \le p^e - \alpha$ and $Q_0 \le p-1$, we get $R' \leq p^{e+1} - \alpha$. Now, notice that $Q_1 \le \frac{\eta + 1 - Q_0}{p} \le \lfloor \frac{p^e-1}{p} \rfloor = p^{e-1} - 1$. Hence,
    \[
        R' \ge R + p^e - 1 - (p-1)p^{e-1} \ge R + p^{e-1} - 1 \ge 1\,.
    \]
        
    Now, assume that $Q_0 = 0$. We thus have
        \[
            p^{e+1} - i + \ps{\bfw}{\bfk} = Q_1(p^{e+1} - 1) + R'
        \]
        where 
        \[
            R' \mydef p^{e+1} - p^e + R - Q_1(p-1).
        \]
    Once again, let us prove that $1 \le R' \le p^{e+1} - \alpha$. It is straightforward to check that $R' \leq p^{e+1}-\alpha$. Moreover, $Q_1 \le \frac{\eta+1}{p} \le p^{e-1}$, leading to
    \[
        R' \ge p^{e+1} - p^e + R - p^{e-1}(p-1) \ge R \ge 1\,.
    \]
\end{proof}

\subsection{On the asymptotic information rate of $\Lift^\eta(\RS_q(d))$ when $q \to \infty$}

In this section, we consider sequences of codes $\Lift^\eta \RS_q(d)$ where $q \ge 2$ varies exponentially (\emph{i.e.} $q = p^e$ with increasing $e$), and where we see $d$ as a function of $q$ such that $d(q) \le q - 2$. Recall that $q$ represents simultaneously the size of the finite field and the square root of the code length. Throughout the section, we will write $q=p^e$.

To our opinion, two cases are of interest: $d = q - \alpha$ where $\alpha \ge 2$ is a fixed integer, and $d = \lfloor \gamma q \rfloor$ where $\gamma \in (0,1)$. In the first case ($d = q - \alpha$) we prove that we obtain $\eta$-lifted codes whose information rate grows to $1$ when $q \to \infty$. In the second case ($d = \lfloor \gamma q \rfloor$) we prove that the sequence of $\eta$-lifted codes admits an asymptotic information rate $R_\gamma > 0$ when $q \to \infty$, meaning that this sequence of codes is asymptotically good and is locally correctable from a constant fraction of errors. In order to prove these results, we look for tight enough lower bounds on the dimension of $\eta$-lifted codes.

\subsubsection{A lower bound for $|D(q, q-\alpha, \eta)|$.}

We first highlight that, for a fixed $\alpha \ge 2$, the degree set $D(q, q-\alpha, \eta)$ of  $\Lift^\eta\RS_q(q-\alpha)$ contains many copies of the degree set of $\WRM^\eta_{p^\epsilon}(p^\epsilon-\alpha-\eta)$, for $\epsilon \leq e$. In terms of codes, it informally means that weighted Reed-Muller codes defined over several fields $\FF_{p^\epsilon}$ for $\epsilon \leq e$, can be embedded in many different manners into $\eta$-lifted codes. This is formalized in the following proposition.

\begin{proposition}
    \label{prop:small-wrm-in-lifts}
    Let $0 \leq \epsilon \leq e$, $\alpha \in [0,p^\epsilon-1]$ and $(i,j) \in  \Deg(\WRM^\eta_{p^\epsilon}(p^\epsilon-\alpha-\eta))$. Then, for every $0 \le a,b \le p^{e-\epsilon}-1$, we have
    \[
        (i+ a p^\epsilon, j + b p^\epsilon) \in D(p^e, p^e - \alpha, \eta)\,.
    \]  
\end{proposition}

\begin{proof}
Assume that $(i,j) \in  \Deg(\WRM^\eta_{p^\epsilon}(p^\epsilon-\alpha-\eta))$. Then 
$i+\eta j \leq p^\epsilon-\alpha-\eta$. We use the characterisation of Proposition~\ref{caracLRS} to prove our result.

Take $\bfk \in \NN^\eta$ such that for all $r \geq 0$,  $\sum_{\ell=1}^\eta k_\ell^{(r)} \leq (j + b p^\epsilon)^{(r)}$.
Then
\[
    \sum_{\ell=1}^\eta k_\ell^{(r)} \leq
        \left\{ \begin{array}{cl}
            j^{(r)}          & \text{ if } r \in [0, \epsilon - 1],  \\
            b^{(r-\epsilon)} & \text{ if } r \in [\epsilon, e - 1], \\
            0                & \text{ if } r \geq e.  
        \end{array} \right.
\]
Our purpose is to bound $\Red{p^e}(i+a p^\epsilon + \ps{\bfw}{\bfk})$. We see that
\[
    i+a p^\epsilon + \ps{\bfw}{\bfk}
    = i+a p^\epsilon + \sum_{\ell=1}^\eta \ell \left( \sum_{r=0}^{\epsilon -1} k_\ell^{(r)} p^r + \sum_{r=\epsilon}^{e-1} k_\ell^{(r)} p^r\right)
    = R_1 + p^\epsilon R_2
\]
where $R_1 \mydef i + \sum_{\ell=0}^\eta \ell \sum_{r=0}^{\epsilon -1} k_\ell^{(r)} p^r$ and $R_2 \mydef a + \sum_{\ell=0}^\eta \ell \sum_{r=\epsilon}^{e-1} k_\ell^{(r)} p^{r-\epsilon}$.

One can check that $R_1 \leq i+\eta j \leq  p^\epsilon - \alpha - \eta$. It remains to deal with $R_2$. Let us write $R_2 = \sum_{r=0}^{e - \epsilon-1} R_2^{(r)} p^r + R_2' p^{e-\epsilon}$ with $R_2' \leq \eta$. Then
\[
    p^\epsilon R_2 = (p^e - 1) R_2'+ R'_2 + \sum_{r=0}^{e-\epsilon-1} R_2^{(r)} p ^{\epsilon+r} \equiv^\star_{p^e} R'_2 + \sum_{r=\epsilon}^{e-1} R_2^{(r)} p^{r}\,.
\]

Therefore,
\[
    i + ap^\epsilon + \langle \bfw, \bfk \rangle \equiv^\star_{p^e} R_1 + R'_2 + \sum_{r=0}^{\epsilon-1} R_2^{(r)} p ^{\epsilon+r} \leq p^\epsilon - \alpha - \eta + \eta + p^\epsilon(p^{e-\epsilon}-1) \leq p^e-\alpha,
\]
which proves that $(i+ ap^\epsilon,j + b p^\epsilon)$ belongs to $D(p^e, p^e - \alpha, \eta)$.
\end{proof}

Notice that $\WRM^\eta_{p^\epsilon}(p^\epsilon-\alpha-\eta) = \{ \bfzero \}$ if $\alpha \ge p^\epsilon$. Therefore let us set $e_\alpha = \lfloor\log_p \alpha \rfloor$ and define
\[
    \mathcal{W}(\epsilon, a, b) \mydef
    \left\{(i + ap^\epsilon, j + b p^\epsilon) \mid (i,j) \in \Deg\WRM^\eta_{p^\epsilon} (p^\epsilon-\alpha-\eta)\right\}
\]
as the degree set of weighted Reed-Muller codes over $\FF_{p^e}$, translated by $(ap^\epsilon,bp^\epsilon)$. Proposition~\ref{prop:small-wrm-in-lifts} ensures that:
\begin{equation}
    \label{eq:big-union}
    D(p^e, p^e - \alpha, \eta) \supset \bigcup_{\epsilon=e_\alpha  +1 }^e \bigcup_{ 0 \leq a,b < p^{e-\epsilon}} \mathcal{W}(\epsilon,a,b).
\end{equation}
Equation~\eqref{eq:big-union} helps us to obtain a first lower bound on the dimension of lifted codes. It is clear that $\mathcal{W}(\epsilon,a,b) \cap \mathcal{W}(\epsilon,a',b') = \varnothing$ if $(a',b') \ne (a,b)$. Unfortunately, the union given in \eqref{eq:big-union} is not disjoint, as illustrated in Figure~\ref{triangles}. The main reason is that $\mathcal{W}(\epsilon,a,b)$ contains a certain number of degree sets of the form $\mathcal{W}(\epsilon',a',b')$, for $\epsilon' < \epsilon$. We compute this precise number in Lemma~\ref{lem:nb-wrm}.

\begin{figure}[ht]
\centering
\includegraphics[scale=0.6]{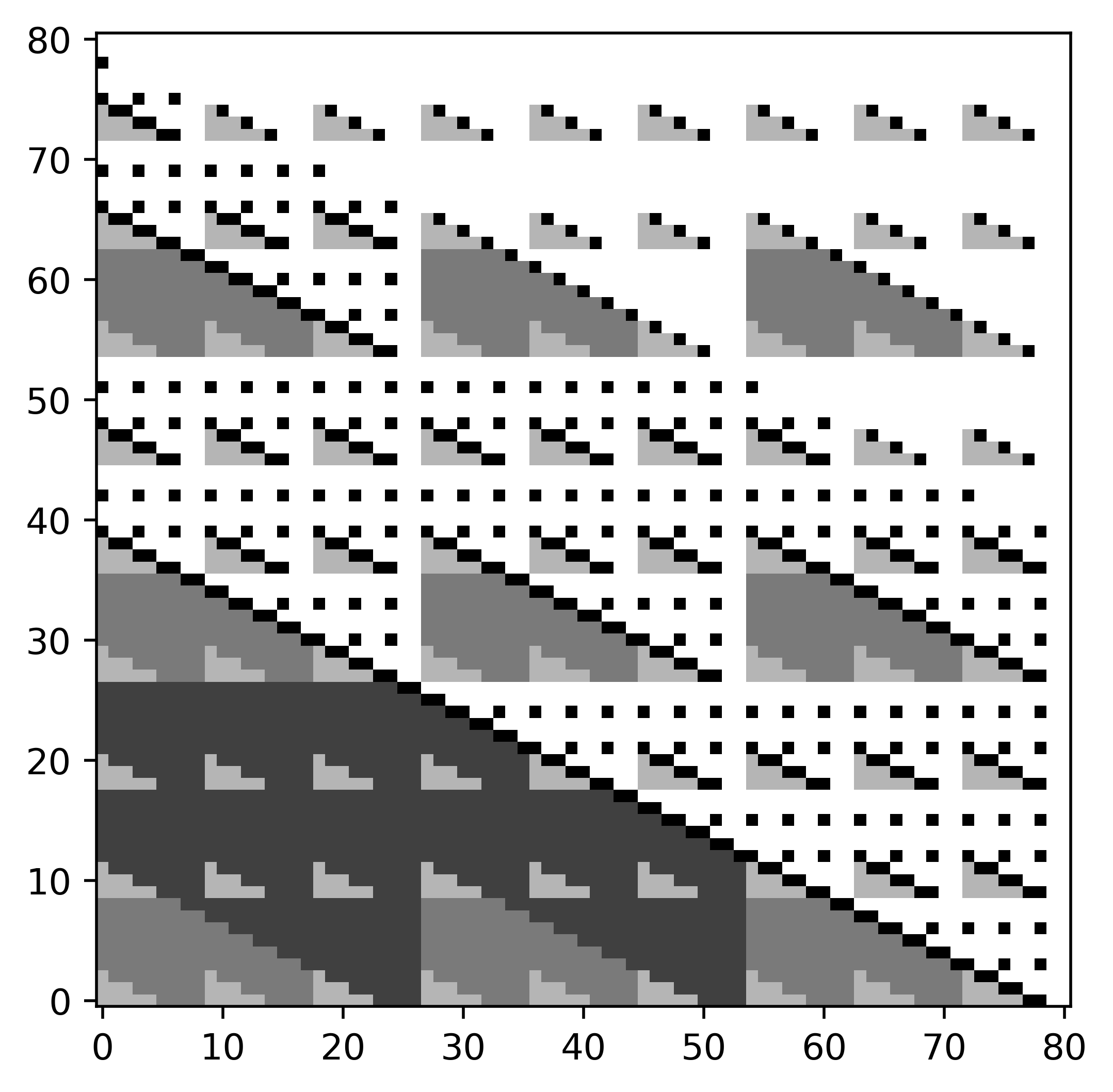}
\caption{Embedding of $\mathcal{W}(\epsilon,a,b) \subset D(3^5, 3^5 - 3, 2)$ with $\epsilon \leq 5$.}
\label{triangles}
\end{figure}

For $m \ge 0$, we set
\begin{equation}
    \label{seqT}
    T_m(p,\eta)=T_m \mydef \left(\left\lfloor\frac{ p^m - 1}{\eta}\right\rfloor +1 \right)\left(p^{m}-\frac{\eta}{2}\left\lfloor\frac{ p^{m} - 1}{\eta}\right\rfloor\right).
\end{equation}
One can check that $T_m$ is a positive integer which counts the number of pairs of non-negative integers $(u,v)$ such that $u + \eta v \le p^m - 1$. 

\begin{lemma}
    \label{lem:nb-wrm}
    Fix $e_\alpha+1 \le \epsilon_1 \leq \epsilon_2 \le e$. Then, for all $0 \le a_2, b_2 < p^{e-\epsilon_2}$, we have:
    \[
    \big| \{ (a_1, b_1) \mid \mathcal{W}(\epsilon_1, a_1, b_1) \subset \mathcal{W}(\epsilon_2, a_2, b_2) \} \big|
    = T_{\epsilon_2-\epsilon_1}\,.
    \]
\end{lemma}
\begin{proof}
    We first notice that ${\cal W}(\epsilon_1, a_1, b_1) \subseteq {\cal W}(\epsilon_2, a_2, b_2)$ if and only if 
    \[
        {\cal W}(\epsilon_1, a_1 - a_2 p^{\epsilon_2-\epsilon_1}, b_1 - b_2p^{\epsilon_2-\epsilon_1}) \subseteq {\cal W}(\epsilon_2, 0, 0)\,.
    \]
    
    Moreover, for $u, v \ge 0$, we see that ${\cal W}(\epsilon_1, u, v) \subseteq {\cal W}(\epsilon_2, 0, 0)$ if and only if for every $i,j \ge 0$, we have
    \[
                i + \eta j \le p^{\epsilon_1} - \alpha - \eta \implies i + u p^{\epsilon_1} + \eta( j + v p^{\epsilon_1}) \le p^{\epsilon_2} - \alpha - \eta\,,
    \]
    which is equivalent to $(u + \eta v) p^{\epsilon_1} \le p^{\epsilon_2}-p^{\epsilon_1}$. It remains to notice that $T_{\epsilon_2 - \epsilon_1}$ counts the number of non-negative integers $u, v$ such that 
    \[
        u + \eta v \le \left\lfloor \frac{p^{\epsilon_2} - p^{\epsilon_1}}{p^{\epsilon_1}} \right\rfloor = p^{\epsilon_2-\epsilon_1} - 1\,.
    \]
\end{proof}

For any $m \in \NN$, we set 
\begin{equation}
\label{eq:defw}
\begin{aligned}
    W_m(\alpha)
        &\mydef |\Deg \WRM^\eta_{p^{m}}(p^m - \alpha - \eta)| = |{\cal W}(m, 0, 0)|\\
        &= \left\lfloor \frac{p^m - \alpha}{\eta} \right\rfloor \left( p^m - \alpha + 1 - \frac{\eta}{2} \left( \left\lfloor \frac{p^m - \alpha}{\eta} + 1  \right\rfloor \right) \right)\,.
\end{aligned}
\end{equation}

Let us also define $N_0 \mydef 1$, and
\begin{equation}
    \label{seqN}
    N_m \mydef p^{2m} - \sum_{\nu=0}^{m-1}N_\nu T_{m-\nu}
\end{equation}
as the number of triangles $\mathcal{W}(e-m, a, b)$ that are not included in any $\mathcal{W}(e-m', a', b')$ with $m' \leq m$. Notice that, equivalently, we have
\begin{equation}
    \label{seqNbis}
    p^{2m} = \sum_{\nu=0}^{m}N_\nu T_{m-\nu}\,.
\end{equation}

\begin{example}
As displayed in Figure \ref{triangles}, for $p=3$ and $\eta=2$, the first terms of the sequence $(N_m)$ are 1, 5, 36, 264.
\end{example}

The following theorem can be proven by a simple counting argument.

\begin{theorem}
    \label{thm:main-thm-lower-bound}
    Fix $\alpha \ge 2$, $\eta \ge 1$ and a prime power $q = p^e$. Let $(W_m(\alpha))_{m \le e}$ and $(N_m)_{m \le e}$ be the sequences defined above. Then, the dimension $|D(q, q-\alpha, \eta)| $ of $\Lift^\eta \RS_q(q-\alpha)$ is lower bounded by
    \[
        \sum_{\epsilon=0 }^{e - e_\alpha - 1} W_{e - \epsilon}(\alpha) N_{\epsilon}\,,
    \]
    where $e_\alpha = \lfloor \log_p \alpha \rfloor$.
\end{theorem}

\subsubsection{Asymptotical behaviour of the sequences $(T_m)$, $(W_m(\alpha))$ and $(N_m)$}

Let us sum up the asymptotics of the sequences introduced in the previous paragraph.

\begin{lemma}\label{lem:asymp}
When $m \rightarrow + \infty$,
\begin{enumerate}
\item $T_m \sim \frac{p^{2m}}{2\eta}$,
\item $W_m(\alpha) \sim T_m$ for any $\alpha \ge 2$.
\end{enumerate}
\end{lemma}

The following technical lemma will be useful in the proof of Theorem \ref{th:inforate1}.

\begin{lemma}
    \label{tech}
    Let $(N_m)$ be the sequence defined in \eqref{seqN}. Then
    \[
        \lim_{m \rightarrow + \infty} \frac{1}{p^{2m}} \sum_{\ell=0}^m N_\ell = 0.
    \]
\end{lemma}
\begin{proof}
Let us first prove that the series $\sum_{\ell \geq 0} \frac{N_\ell}{p^{2\ell}}$ is convergent.
Fix $\delta > 0$.

By Lemma \ref{lem:asymp}, $T_m \sim \frac{p^{2m}}{2\eta}$. Hence there exists $L \in \NN$ such that for any $\ell \geq L$, $p^{2\ell} \leq (2\eta+\delta) T_\ell$. Therefore, using \eqref{seqN}, we get
\[
\begin{aligned}
    \sum_{\ell=0}^m \frac{N_\ell}{p^{2\ell}}
    &= \sum_{\ell=0}^{m-L} \frac{N_\ell}{p^{2\ell}} + \sum_{\ell=m-L+1}^m \frac{N_\ell}{p^{2\ell}}\\
    &\le \frac{1}{p^{2m}} \sum_{\ell=0}^{m-L} N_\ell p^{2(m-\ell)} + \sum_{\ell=m-L+1}^m \frac{N_\ell}{p^{2\ell}}\\
    &\le \frac{(2 \eta + \delta)}{p^{2m}} \sum_{\ell=0}^{m-L} N_\ell T_{m-\ell} + \sum_{\ell=m-L+1}^m \frac{N_\ell}{p^{2\ell}}\,,
\end{aligned}
\]
since $m - \ell \ge L  \iff  \ell \leq m-L$.

Notice that all the terms of the first sum are non-negative. Hence by \eqref{seqNbis}, we have $\sum_{\ell=0}^{m-L} N_\ell T_{m-\ell} \le p^{2m}$, leading to
\[
\sum_{\ell=0}^m \frac{N_\ell}{p^{2\ell}} \leq  (2\eta +\delta) +  \sum_{\ell=m-L+1}^m \frac{N_\ell}{p^{2\ell}}.
\]
It remains to notice that the right handside sum is finite, and each summand $N_\ell/p^{2\ell}$ is trivially bounded by $1$. Therefore $\sum_{\ell \ge 0} N_\ell / p^{2\ell}$ is convergent.

Denote by $S$ its limit. We know there exists $M \in \NN$ such that, for any $m \geq M$ it holds that
\[
   \left| S - \sum_{\ell=0}^m \frac{N_\ell}{p^{2\ell}} \right| \le \delta\,.
\]
As a consequence, $\sum_{\ell=M+1}^m N_\ell/p^{2\ell} \le 2 \delta$ and since $\sum_{\ell=0}^M N_\ell/p^{2\ell} \le S$, we get 
\[
    \frac{1}{p^{2m}}\sum_{\ell=0}^m N_\ell = \sum_{\ell=0}^M \frac{N_\ell}{p^{2\ell}} \frac{1}{p^{2(m-\ell)}} + \sum_{\ell=M+1}^m \frac{N_\ell}{p^{2 \ell}} \leq \frac{S}{p^{2(m-M)}} + 2\delta,
    \]
which concludes the proof.
\end{proof}

\subsubsection{Asymptotics of the rate of $\Lift^\eta \RS_q(q-\alpha)$ when $q \to \infty$ and $\alpha$ is fixed}

\begin{theorem}\label{th:inforate1}
    Let $\alpha \ge 2$, $\eta \ge 1$ and $p$ be a prime number. Define $e_\alpha = \lfloor \log_p \alpha \rfloor$, and consider the sequence of codes $\calC_e = \Lift^\eta \RS_{p^e}(p^e - \alpha)$, for $e \ge e_\alpha$. Then, the information rate $R_e$ of $\calC_e$ approaches $1$ when $e \to \infty$.
\end{theorem}
\begin{proof}
    By Lemma \ref{lem:asymp}, $W_m(\alpha) \sim_{m \rightarrow +\infty} T_m$. Fix $\delta > 0$ and let $M \ge e_\alpha$ such that for every $m \geq M$, $W_m(\alpha) \geq (1 - \delta) T_m$.
    
    Using Theorem~\ref{thm:main-thm-lower-bound}, we thus get 
    \[
    \begin{aligned}
        |D(p^e, p^e-\alpha, \eta)|
        &\ge \sum_{\epsilon=0 }^{e-e_\alpha-1} W_{e-\epsilon}(\alpha) N_{\epsilon}\\
        &\geq (1 - \delta)\sum_{\epsilon=0 }^{e-M} T_{e-\epsilon} N_{\epsilon} + \sum_{\epsilon=e-M+1 }^{e-e_\alpha-1} W_{e-\epsilon}(\alpha)  N_{\epsilon}\\
        & \geq (1 - \delta)\left(p^{2e} - \sum_{\epsilon=e-M+1 }^e T_{e-\epsilon} N_{\epsilon} \right) + \sum_{\epsilon=e-M+1 }^{e-e_\alpha-1} W_{e-\epsilon}(\alpha)  N_{\epsilon}\\
        & \geq  (1 - \delta)\left(p^{2e} - T_{M-1} \sum_{\epsilon=e-M+1 }^e N_{\epsilon} \right) + W_{M-1}(\alpha)  \sum_{\epsilon=e-M+1 }^{e-e_\alpha-1} N_{\epsilon}\,.
        \end{aligned}
        \]
Then, by Lemma \ref{tech}, both terms $\sum_{\epsilon=e-M+1 }^e N_{\epsilon}/p^{2e}$ and $\sum_{\epsilon=e-M+1 }^{e-e_\alpha-1} N_{\epsilon}/p^{2e}$ vanish when $e \to \infty$. Hence we get
\[
    R_e = \frac{|D(q, q-\alpha, \eta)|}{p^{2e}} \rightarrow 1\,.
\]
\end{proof}

\begin{example}
    Let us give some numerical computations of the dimension and information rate of $\Lift^\eta \RS_{p^e}(p^e-\alpha)$ illustrating Theorem~\ref{th:inforate1}.
    \[
    \begin{array}{c|c|c|c|c|c|c}
        p & \eta & \alpha & e & n = p^{2e} & k = |D(p^e, p^{e-c}, \eta)| & R = k/n \\
        \hline
        \hline
        \multirow{8}{*}{2} & \multirow{8}{*}{2} & \multirow{8}{*}{2} & 3 & 64 & 25 & 0.3906 \\
         &  &  &  4 & 256 & 121 & 0.4727 \\
         &  &  &  5 & 1024 & 561 & 0.5479 \\
         &  &  &  6 & 4096 & 2513 & 0.6135 \\
         &  &  &  7 & 16384 & 10977 & 0.6700 \\
         &  &  &  8 & 65536 & 47073  & 0.7183 \\
         &  &  &  9 & 262144 & 199105  & 0.7595 \\
         &  &  & 10 & 1048576 & 833345  & 0.7947 \\
        \hline
        \hline
        \multirow{5}{*}{2} & \multirow{5}{*}{2} & \multirow{5}{*}{16} & 6 & 4096 & 781 & 0.1907 \\
         &  &  &  7 & 16384 & 4944 & 0.3018 \\
         &  &  &  8 & 65536 & 26335 & 0.4018 \\
         &  &  &  9 & 262144 & 128142 & 0.4888 \\
         &  &  & 10 & 1048576 & 590885 & 0.5635 \\
        \hline
        \hline
        \multirow{5}{*}{2} & \multirow{5}{*}{4} & \multirow{5}{*}{2} &  3 & 64 & 16 & 0.2500 \\
         &  &  &  4 & 256 & 71 & 0.2773 \\
         &  &  &  5 & 1024 & 331 & 0.3232 \\
         &  &  &  6 & 4096 & 1506 & 0.3677 \\
         &  &  &  7 & 16384 & 6749 & 0.4119 \\
        \hline
        \hline
    \end{array}
    \]

    In Figure~\ref{fig:d= q - alpha}, we also represent the degree sets of $\Lift^2 \RS_{2^e}(2^e-\alpha)$ for $\alpha=3$ and $e \in \{7, 8, 9, 10\}$.
\end{example}

\begin{figure}[ht]
\centering
\begin{subfigure}[t]{0.25\textwidth}
\centering
\includegraphics[width=\textwidth]{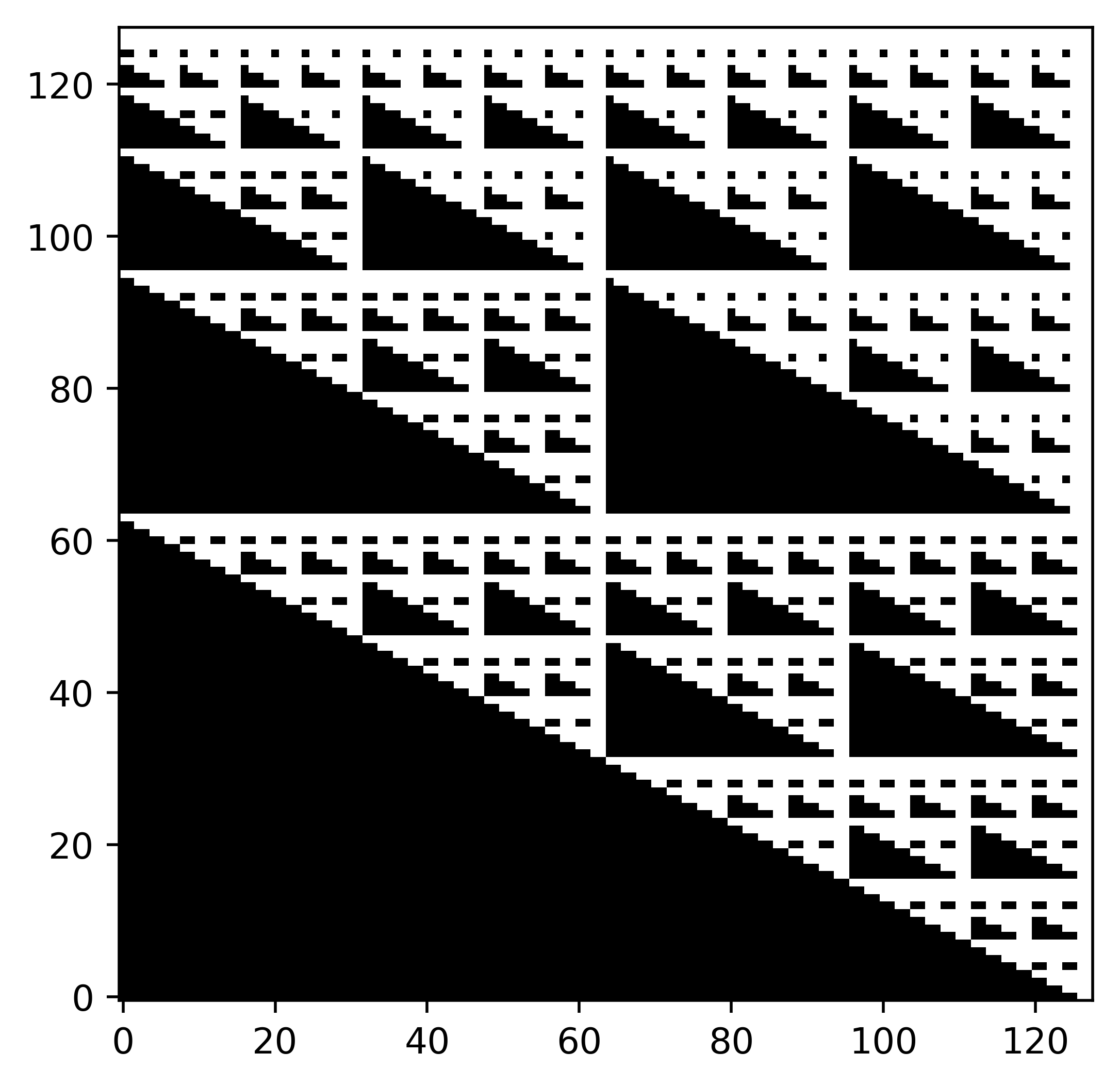}
\subcaption{$e=7$}
\end{subfigure}%
\begin{subfigure}[t]{0.25\textwidth}
\centering
\includegraphics[width=\textwidth]{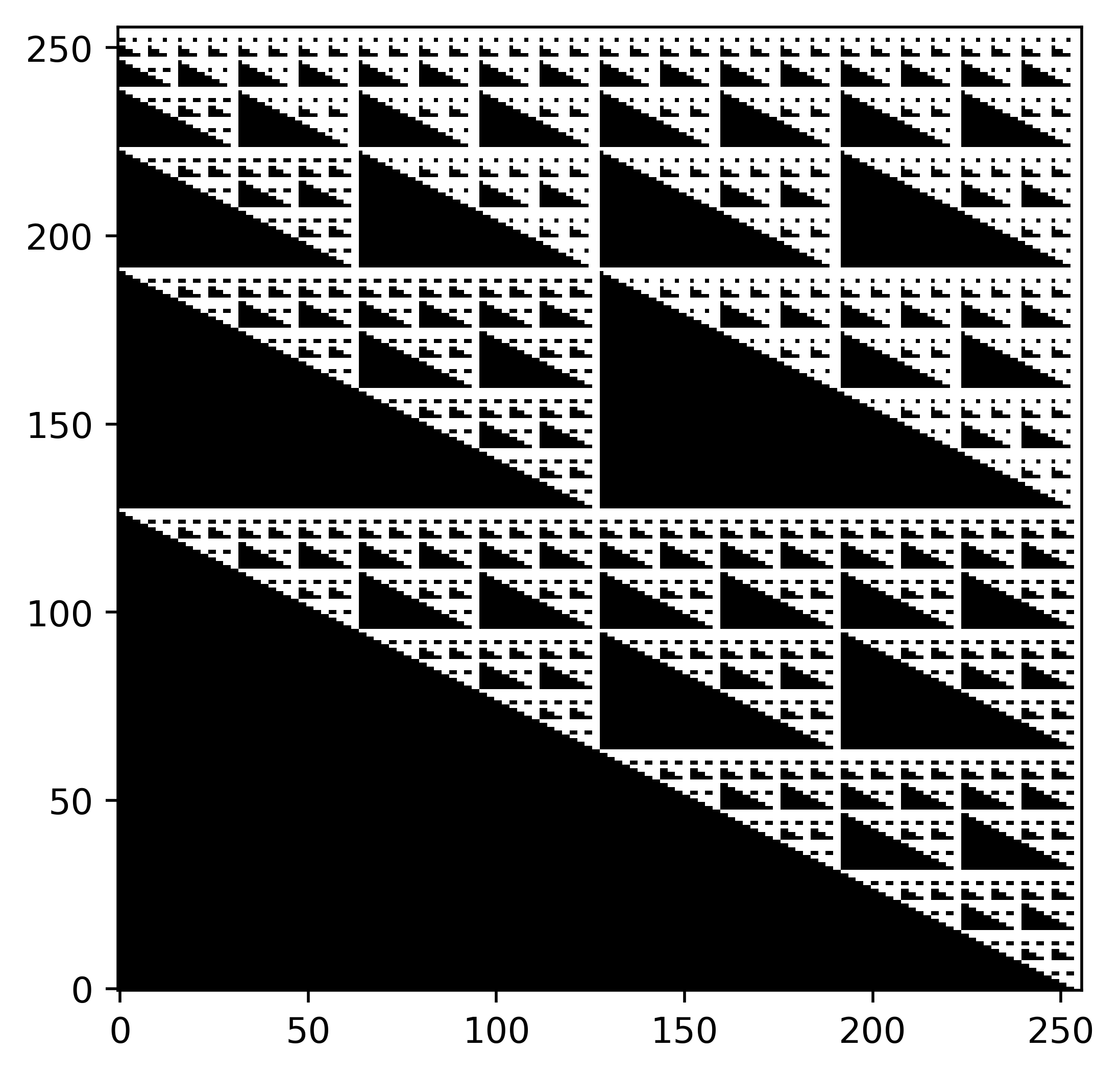}
\subcaption{$e=8$}
\end{subfigure}%
\begin{subfigure}[t]{0.25\textwidth}
\centering
\includegraphics[width=\textwidth]{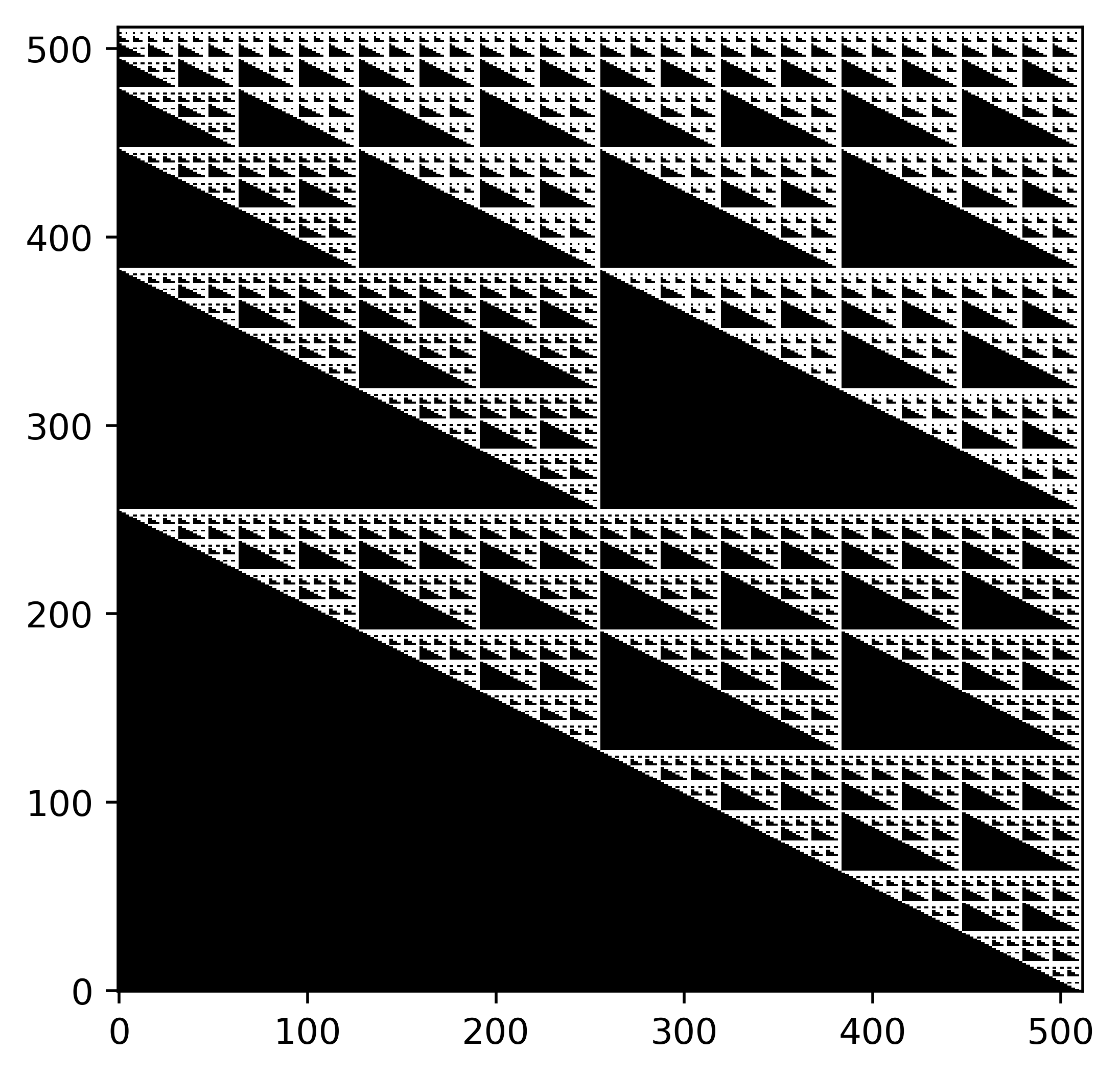}
\subcaption{$e=9$}
\end{subfigure}%
\begin{subfigure}[t]{0.25\textwidth}
\centering
\includegraphics[width=\textwidth]{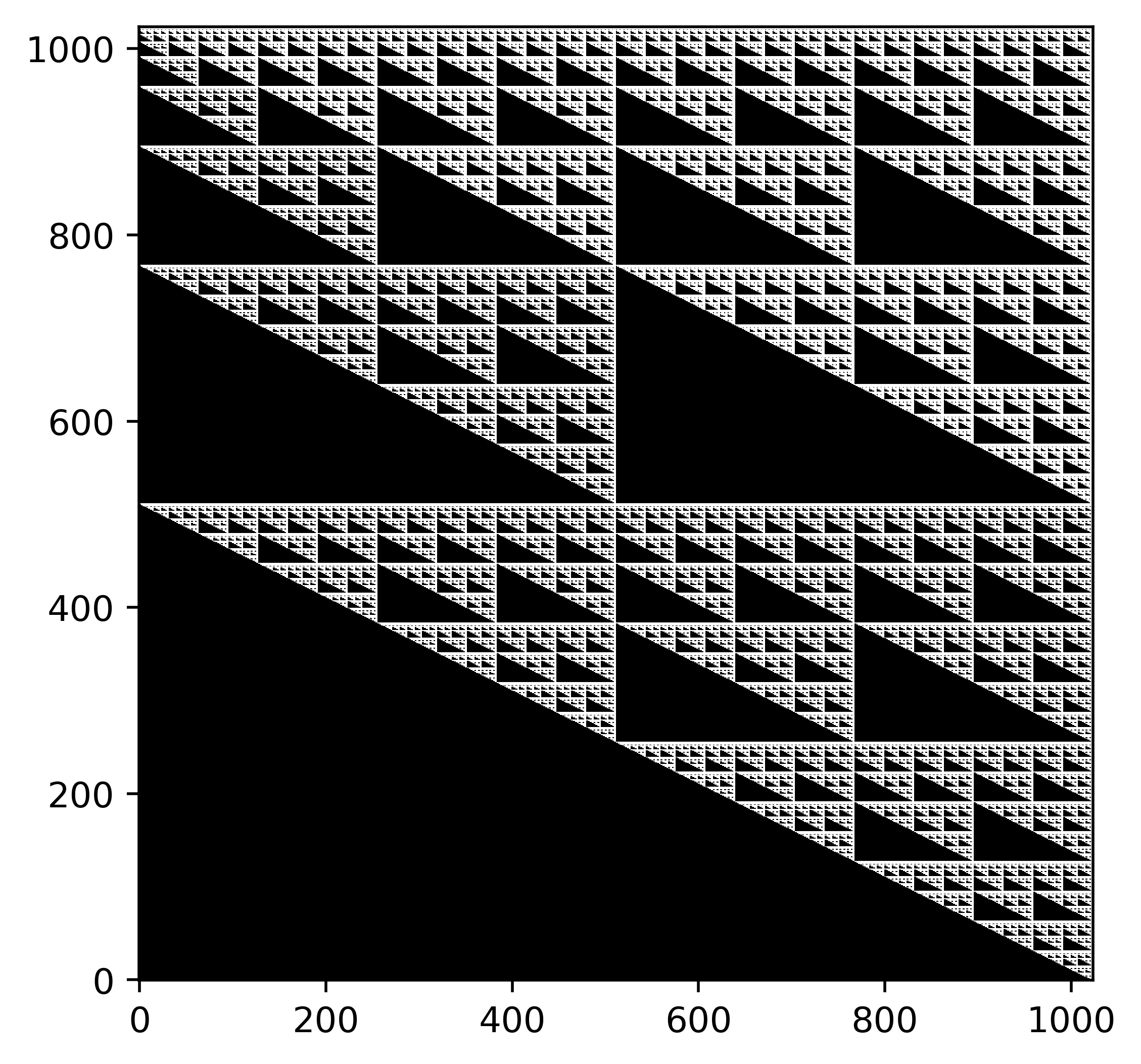}
\subcaption{$e=10$}
\end{subfigure}
\caption{Representation of the degree set of $\Lift^2 \RS_{2^e}(2^e-\alpha)$ for $\alpha=3$ and different values of $e$.}\label{fig:d= q - alpha}
\end{figure}

\subsubsection{Asymptotics of the rate of $\Lift^\eta \RS_q(\lfloor \gamma q \rfloor)$ when $q \to \infty$ and $\gamma$ is fixed}

\begin{theorem}
    \label{thm:asymp-good}
    Let $c \ge 1$, $\eta \ge 1$ and $p$ be a prime number.  Define $\gamma = 1 - p^{-c}$, and consider the sequence of codes $\calC_e = \Lift^\eta \RS_{p^e}(\gamma p^e)$, for $e \ge c+1$. Then, the information rate $R_e$ of $\calC_e$ satisfies:
    \[
        \lim_{e \to \infty} R_e \ge \frac{1}{2\eta}\sum_{\epsilon = 0}^{c-1} (p^{-\epsilon} - p^{-c})^2 N_\epsilon\,.
    \]
\end{theorem}

\begin{proof}
    By Proposition \ref{prop:small-wrm-in-lifts},
    \[
        |D(p^e,p^e-p^{e-c},\eta)| \geq \sum_{\epsilon =0 }^{c-1} W_{e-\epsilon}(p^{e-c}) N_\epsilon\,.
    \]
    Moreover, using \eqref{eq:defw}, for every fixed $\epsilon \le c-1$ we have 
    \[
        \lim_{e \to \infty} W_{e-\epsilon}(p^{e-c})= p^{2e}\frac{(p^{-\epsilon} - p^{-c})^2}{2 \eta}\,.
    \]
    Then 
    \[
        \lim_{e \rightarrow \infty} R_e \geq \frac{1}{2\eta} \sum_{\epsilon = 0}^{c-1} (p^{-\epsilon} - p^{-c})^2 N_\epsilon\,.
    \]
\end{proof}

\begin{example}
    Let us give some numerical computations, illustrating the tightness of the bound given in Theorem~\ref{thm:asymp-good}.
    \[
    \begin{array}{c|c|c|c|c|c|c}
        p & \eta & c & e & n = p^{2e} & k = |D(p^e, p^{e-c}, \eta)| & R = k/n \\
        \hline
        \hline
        \multirow{7}{*}{2} & \multirow{7}{*}{2} & \multirow{7}{*}{4} & 5 & 1024 & 561     & 0.5479 \\
         &  &  &  6 & 4096 & 1861    & 0.4543 \\
         &  &  &  7 & 16384 & 6843    & 0.4177 \\
         &  &  &  8 & 65536 & 26335   & 0.4018 \\
         &  &  &  9 & 262144 & 103431  & 0.3946 \\
         &  &  & 10 & 1048576 & 410071  & 0.3911 \\
        \cline{4-7}
        & & & \multicolumn{3}{c|}{\text{lower bound on the asymptotic rate}} & 0.3877 \\
        \hline
        \hline
        \multirow{5}{*}{2} & \multirow{5}{*}{2} & \multirow{5}{*}{6} &  7 & 16384 &    10977 & 0.6700 \\
         &  &  &  8 & 65536 & 39431 & 0.6017 \\
         &  &  &  9 & 262144 & 150729 & 0.5750 \\
         &  &  & 10 & 1048576 & 590885 & 0.5635 \\
        \cline{4-7}
        & & & \multicolumn{3}{c|}{\text{lower bound on the asymptotic rate}} & 0.5533 \\
        \hline
        \hline
        \multirow{5}{*}{2} & \multirow{5}{*}{4} & \multirow{5}{*}{3} &  4 & 256 & 71    & 0.2773  \\
         &  &  &  5 & 1024 & 205 & 0.2002 \\
         &  &  &  6 & 4096 & 699 & 0.1707 \\
         &  &  &  7 & 16384 & 2587 & 0.1579 \\
        \cline{4-7}
        & & & \multicolumn{3}{c|}{\text{lower bound on the asymptotic rate}} &  0.1465 \\
        \hline
        \hline
        \multirow{4}{*}{5} & \multirow{4}{*}{2} & \multirow{4}{*}{2} &  3 & 15625 & 5789 & 0.3705 \\
         &  &  &  4 & 390625 & 132109 & 0.3382 \\
         &  &  &  5 & 9765625 & 3259709 & 0.3338 \\
        \cline{4-7}
        & & & \multicolumn{3}{c|}{\text{lower bound on the asymptotic rate}} & 0.3328 \\    \hline
        \hline
    \end{array}
    \]
    
    In Figure \ref{fig:d=gamma q}, we also represent the degree sets $D(2^e,2^e-2^{e-c},\eta)$ for  $p=2$, $\eta=2$, $c=4$ and $e \in \{ 5,6,7,8 \}$.
\end{example}

\begin{figure}[!ht]
\centering
\begin{subfigure}[t]{0.25\textwidth}
\centering
\includegraphics[width=\textwidth]{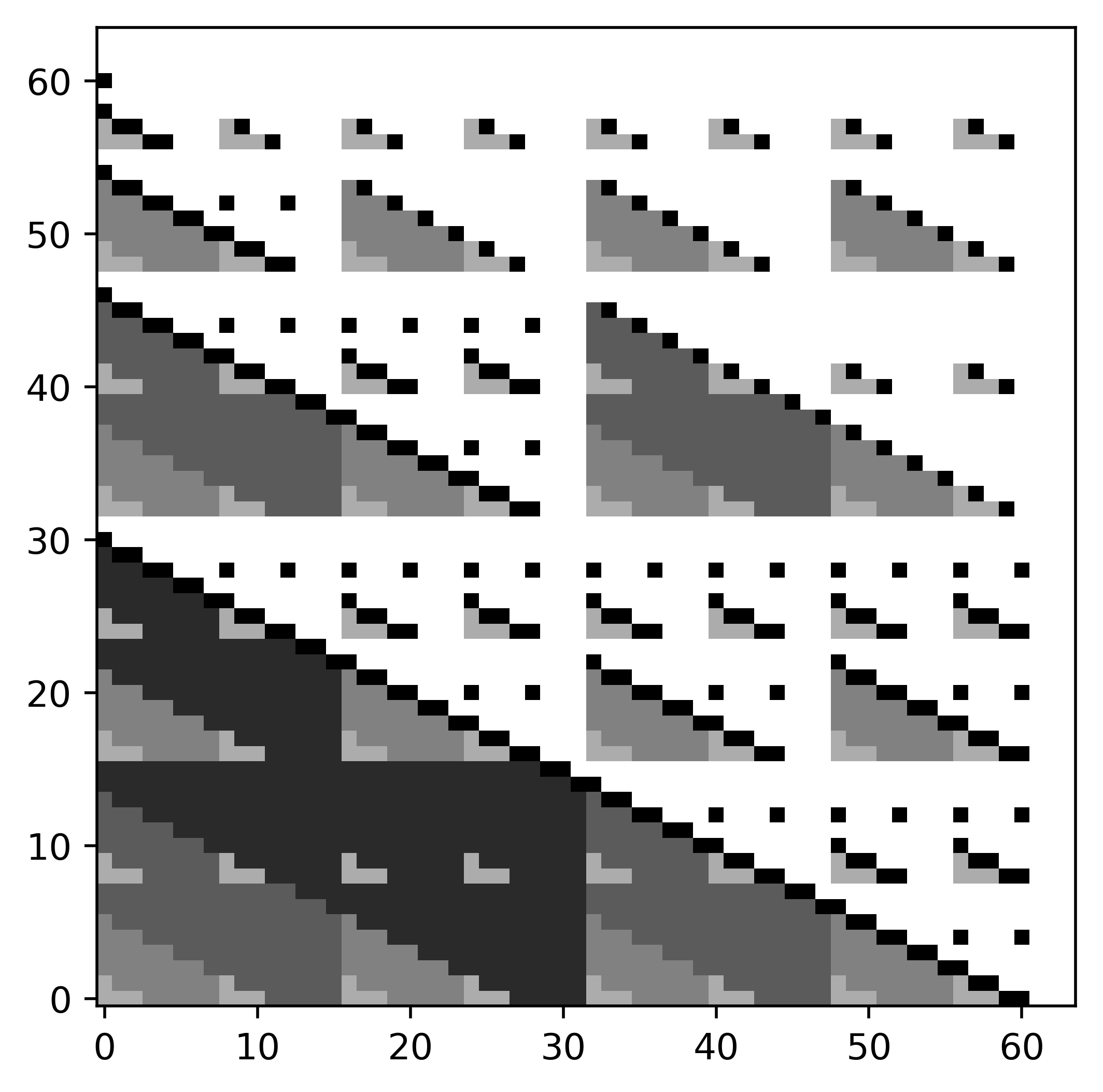}
\subcaption{$e=5$}
\end{subfigure}%
\begin{subfigure}[t]{0.25\textwidth}
\centering
\includegraphics[width=\textwidth]{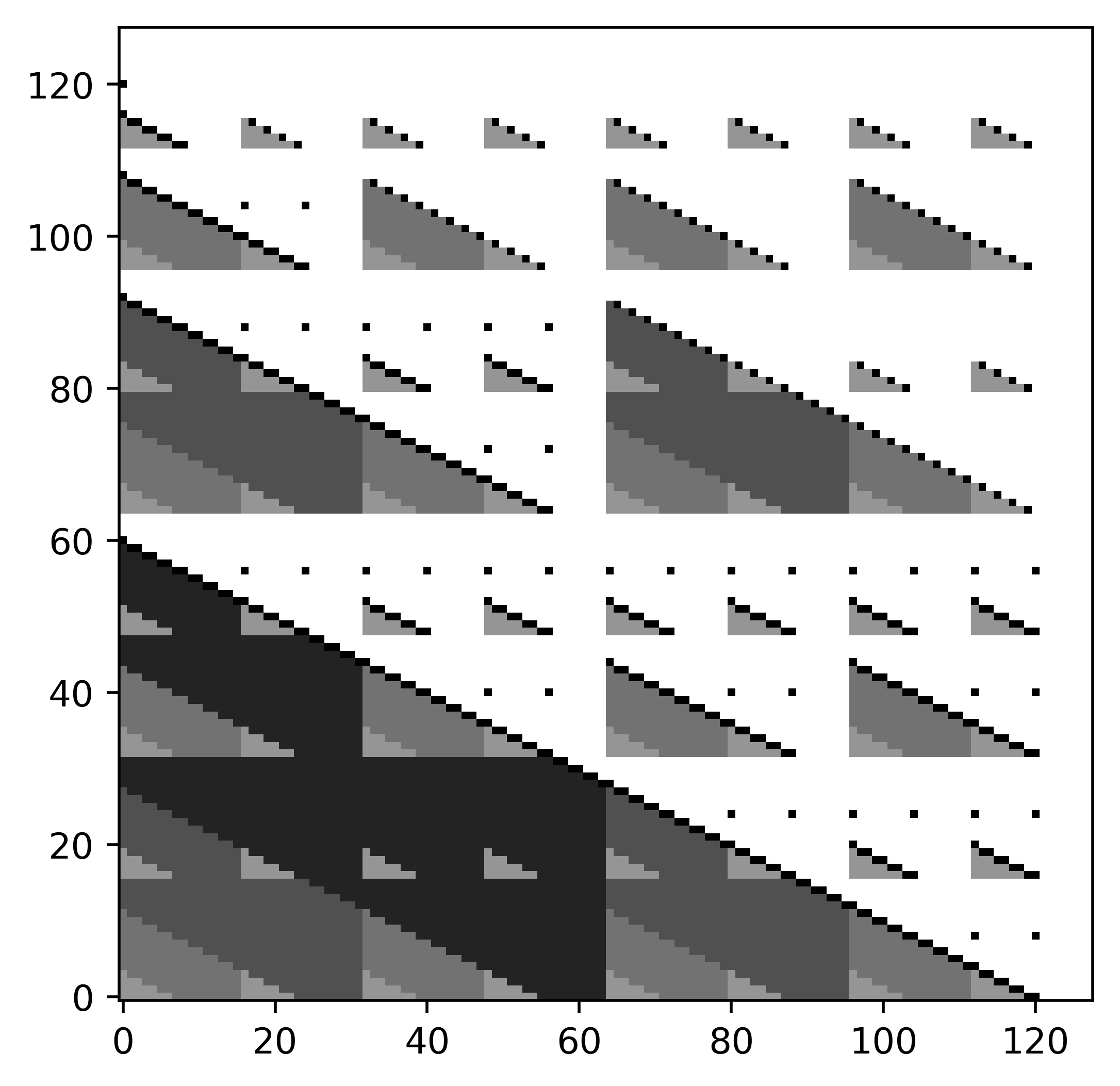}
\subcaption{$e=6$}
\end{subfigure}%
\begin{subfigure}[t]{0.25\textwidth}
\centering
\includegraphics[width=\textwidth]{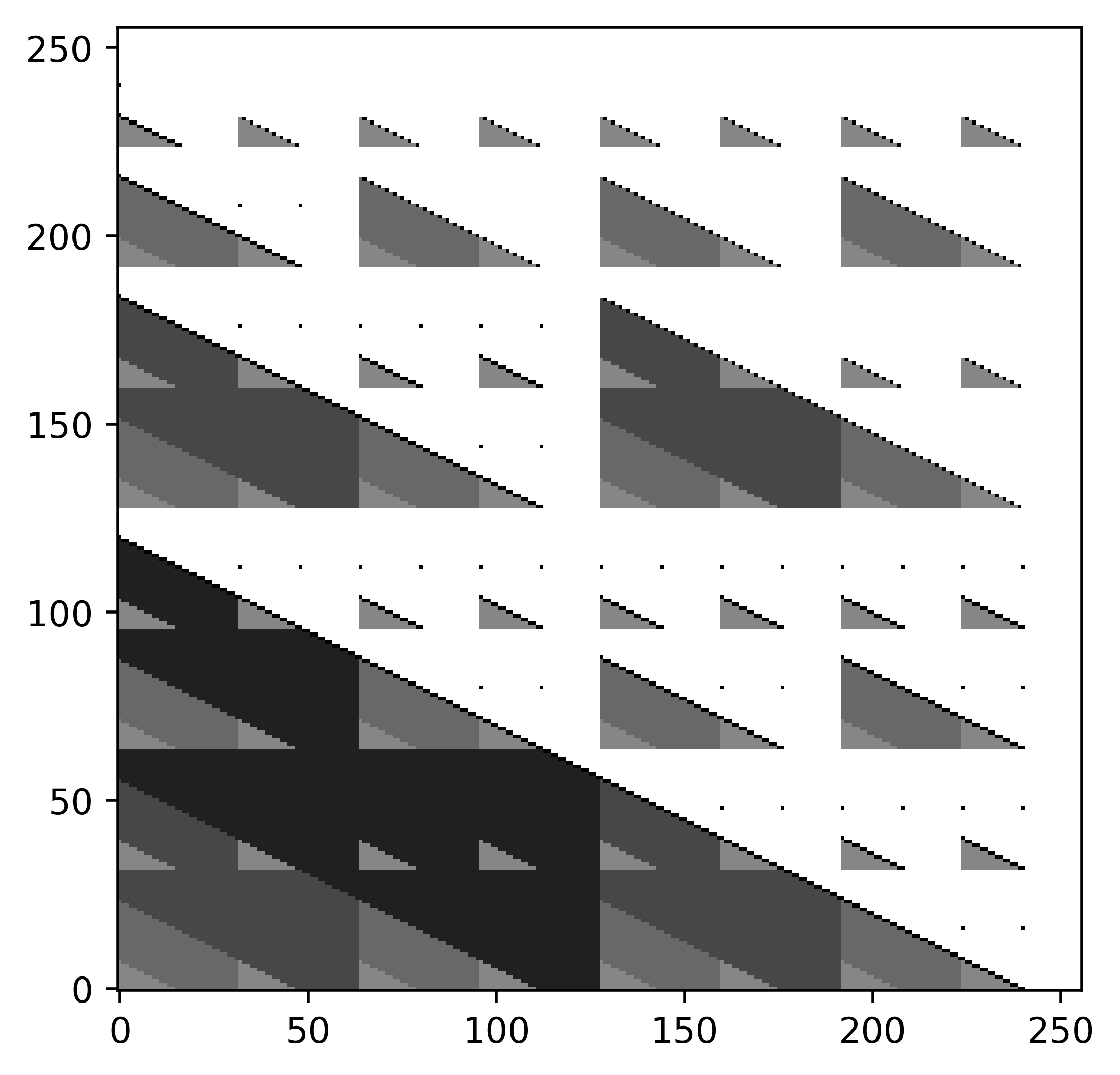}
\subcaption{$e=7$}
\end{subfigure}%
\begin{subfigure}[t]{0.25\textwidth}
\centering
\includegraphics[width=\textwidth]{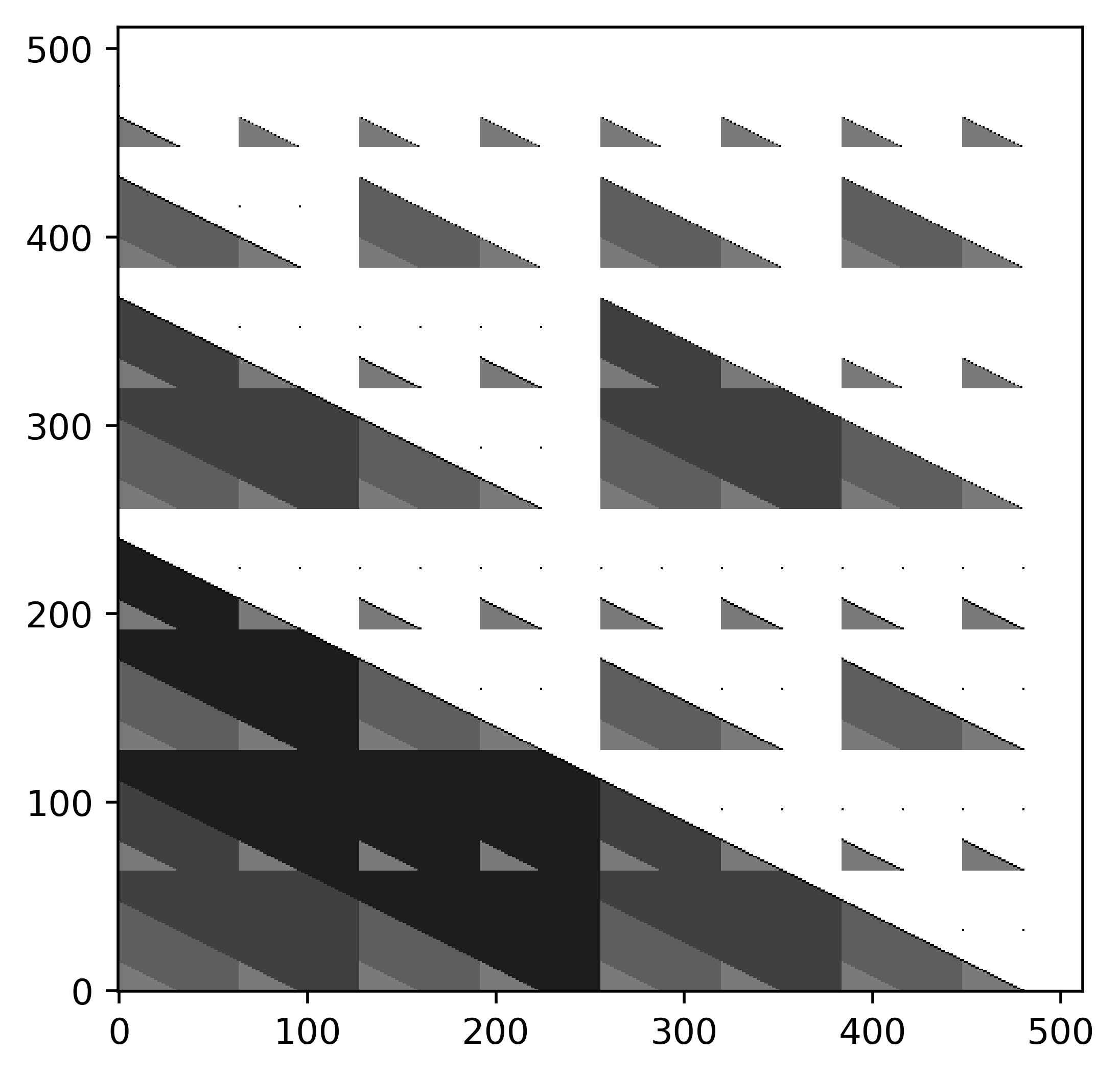}
\subcaption{$e=8$}
\end{subfigure}
\caption{Representation of the degree set of $\Lift^2 \RS_{2^e}(2^e-2^{e-c})$ for $c=4$ and different values of $e$. Note that in each case, the number of differents shades of grey is constant and equal to $c$.}\label{fig:d=gamma q}
\end{figure}


\section*{Acknowledgements}

Part of this work was done while the first author was affiliated to LIX, École Polytechnique, Inria \& CNRS UMR 7161, University Paris-Saclay, Palaiseau, France.  The first author is now funded by the French \emph{Direction Générale de l'Armement}, through the \emph{Pôle d'excellence cyber}. This work was also funded in part by the grant ANR-15-CE39-0013-01 \enquote{Manta} from the French National Research Agency, which gave the authors the opportunity to work together.

\newcommand{\etalchar}[1]{$^{#1}$}


\begin{thebibliography}{FHGHK17}

\bibitem[ACG{\etalchar{+}}17]{AubryCGLOR17}
Yves Aubry, Wouter Castryck, Sudhir~R. Ghorpade, Gilles Lachaud, Michael~E.
  O'Sullivan, and Samrith Ram.
\newblock {Hypersurfaces in Weighted Projective Spaces Over Finite Fields with
  Applications to Coding Theory}.
\newblock In Everett~W. Howe, Kristin~E. Lauter, and Judy~L. Walker, editors,
  {\em Algebraic Geometry for Coding Theory and Cryptography}, pages 25--61,
  Cham, 2017. Springer International Publishing.

\bibitem[ALS14]{AugotLS14}
Daniel Augot, Fran{\c{c}}oise Levy{-}dit{-}Vehel, and Abdullatif Shikfa.
\newblock A storage-efficient and robust private information retrieval scheme
  allowing few servers.
\newblock In Dimitris Gritzalis, Aggelos Kiayias, and Ioannis~G. Askoxylakis,
  editors, {\em Cryptology and Network Security - 13th International
  Conference, {CANS} 2014, Heraklion, Crete, Greece, October 22-24, 2014.
  Proceedings}, volume 8813 of {\em Lecture Notes in Computer Science}, pages
  222--239. Springer, 2014.

\bibitem[BIKR02]{BeimelIKR02}
Amos Beimel, Yuval Ishai, Eyal Kushilevitz, and Jean{-}Fran{\c{c}}ois Raymond.
\newblock {Breaking the $O(n^{1/(2k-1)})$ Barrier for Information-Theoretic
  Private Information Retrieval}.
\newblock In {\em 43rd Symposium on Foundations of Computer Science ({FOCS}
  2002), 16-19 November 2002, Vancouver, BC, Canada, Proceedings}, pages
  261--270. {IEEE} Computer Society, 2002.

\bibitem[BS02]{BeimelS02}
Amos Beimel and Yoav Stahl.
\newblock Robust information-theoretic private information retrieval.
\newblock In Stelvio Cimato, Clemente Galdi, and Giuseppe Persiano, editors,
  {\em Security in Communication Networks, Third International Conference,
  {SCN} 2002, Amalfi, Italy, September 11-13, 2002. Revised Papers}, volume
  2576 of {\em Lecture Notes in Computer Science}, pages 326--341. Springer,
  2002.

\bibitem[CGKS95]{ChorGKS95}
Benny Chor, Oded Goldreich, Eyal Kushilevitz, and Madhu Sudan.
\newblock {Private Information Retrieval}.
\newblock In {\em 36th Annual Symposium on Foundations of Computer Science,
  Milwaukee, Wisconsin, 23-25 October 1995}, pages 41--50. {IEEE} Computer
  Society, 1995.

\bibitem[DG16]{DvirG16}
Zeev Dvir and Sivakanth Gopi.
\newblock {2-Server PIR with Subpolynomial Communication}.
\newblock {\em J. {ACM}}, 63(4):39:1--39:15, 2016.

\bibitem[Efr12]{Efremenko12}
Klim Efremenko.
\newblock {3-Query Locally Decodable Codes of Subexponential Length}.
\newblock {\em {SIAM} J. Comput.}, 41(6):1694--1703, 2012.

\bibitem[FHGHK17]{FreijHollantiGHK17}
Ragnar Freij-Hollanti, Oliver~W. Gnilke, Camilla Hollanti, and David~A. Karpuk.
\newblock {Private Information Retrieval from Coded Databases with Colluding
  Servers}.
\newblock {\em SIAM J. Appl. Algebra Geometry}, 1(1):647--664, 2017.

\bibitem[GKS13]{GuoKS13}
Alan Guo, Swastik Kopparty, and Madhu Sudan.
\newblock {New Affine-Invariant Codes from Lifting}.
\newblock In Robert~D. Kleinberg, editor, {\em Innovations in Theoretical
  Computer Science, {ITCS} '13, Berkeley, CA, USA, January 9-12, 2013}, pages
  529--540. {ACM}, 2013.

\bibitem[GT13]{GT}
Olav Geil and Casper Thomsen.
\newblock Weighted {R}eed-{M}uller codes revisited.
\newblock {\em Des. Codes Cryptogr.}, 66(1-3):195--220, 2013.

\bibitem[Guo16]{Guo16}
Alan Guo.
\newblock {High-Rate Locally Correctable Codes via Lifting}.
\newblock {\em {IEEE} Trans. Information Theory}, 62(12):6672--6682, 2016.

\bibitem[KRGiA17]{KumarRA17}
Siddhartha Kumar, Eirik Rosnes, and Alexander Graell~i Amat.
\newblock {Private Information Retrieval in Distributed Storage Systems using
  an Arbitrary Linear Code}.
\newblock In {\em 2017 {IEEE} International Symposium on Information Theory,
  {ISIT} 2017, Aachen, Germany, June 25-30, 2017}, pages 1421--1425. {IEEE},
  2017.

\bibitem[KT00]{KatzT00}
Jonathan Katz and Luca Trevisan.
\newblock {On the Efficiency of Local Decoding Procedures for Error-Correcting
  Codes}.
\newblock In F.~Frances Yao and Eugene~M. Luks, editors, {\em Proceedings of
  the Thirty-Second Annual {ACM} Symposium on Theory of Computing, May 21-23,
  2000, Portland, OR, {USA}}, pages 80--86. {ACM}, 2000.

\bibitem[Lav18a]{JLthese}
Julien Lavauzelle.
\newblock {\em {Codes with locality: constructions and applications to
  cryptographic protocols}}.
\newblock Phd thesis, {Universit{\'e} Paris-Saclay}, 2018.

\bibitem[Lav18b]{Lavauzelle18b}
Julien Lavauzelle.
\newblock {Lifted Projective Reed-Solomon Codes}.
\newblock {\em Designs, Codes and Cryptography}, 2018.
\newblock To appear.

\bibitem[Luc78]{Lucas78}
\'Edouard Lucas.
\newblock {Th\'eorie des Fonctions Num\'eriques Simplement P\'eriodiques}.
\newblock {\em American Journal of Mathematics}, 1(3):197--240, 1878.

\bibitem[S{\o}r92]{Sor}
Anders~Bj{\ae}rt S{\o}rensen.
\newblock Weighted {R}eed-{M}uller codes and algebraic-geometric codes.
\newblock {\em IEEE Trans. Inform. Theory}, 38(6):1821--1826, 1992.

\bibitem[SRR14]{ShahRR14}
Nihar~B. Shah, K.~V. Rashmi, and Kannan Ramchandran.
\newblock {One Extra Bit of Download Ensures Perfectly Private Information
  Retrieval}.
\newblock In {\em 2014 {IEEE} International Symposium on Information Theory,
  Honolulu, HI, USA, June 29 - July 4, 2014}, pages 856--860. {IEEE}, 2014.

\bibitem[TGR18]{TajeddineGR18}
Razan Tajeddine, Oliver~W. Gnilke, and Salim~El Rouayheb.
\newblock Private information retrieval from {MDS} coded data in distributed
  storage systems.
\newblock {\em {IEEE} Trans. Information Theory}, 64(11):7081--7093, 2018.

\bibitem[Yek08]{Yekhanin08}
Sergey Yekhanin.
\newblock {Towards 3-query Locally Decodable Codes of Subexponential Length}.
\newblock {\em J. {ACM}}, 55(1):1:1--1:16, 2008.

\bibitem[Yek12]{Yekhanin12}
Sergey Yekhanin.
\newblock {Locally Decodable Codes}.
\newblock {\em Foundations and Trends in Theoretical Computer Science},
  6(3):139--255, 2012.

\end{thebibliography}
\end{document}